\documentclass[11pt, letterpaper]{article}
\usepackage[utf8]{inputenc}
\usepackage{enumitem}
\usepackage{pgfplots}
\pgfplotsset{compat=1.15}
\usetikzlibrary{arrows}
\usepackage[T1]{fontenc}
\usepackage{lmodern}

\usepackage{booktabs, multirow, makecell, array}
\usepackage{colortbl}

\usepackage{amsthm}

\usepackage[english]{babel}
\usepackage{titling}
\usepackage{a4wide}

\usepackage{tikz}
\usepackage{authblk}

\usepackage{booktabs}
\usepackage{thm-restate}
\usepackage{soul}
\usepackage{natbib}

\usepackage[utf8]{inputenc} 
\usepackage[T1]{fontenc}    
\usepackage[colorlinks=true, linkcolor=blue, urlcolor=blue, citecolor=blue]{hyperref}

\usepackage{url}            
\usepackage{booktabs}       
\usepackage{amsfonts}       
\usepackage{nicefrac}       
\usepackage{microtype}      
\usepackage{xcolor}         
\usepackage{csquotes}

\usepackage{amsmath}
\usepackage{amssymb}
\usepackage{thmtools}

\usepackage{graphics}
\usepackage{caption}
\usepackage{pgf}
\pdfpageattr {/Group << /S /Transparency /I true /CS /DeviceRGB>>}

\usepackage[ruled]{algorithm2e}

\usepackage{refcount}
\usepackage{fmtcount}
\usepackage{mleftright}

\usepackage{multirow}
\usepackage{tabularx}

\usetikzlibrary{calc}
\usepackage{multirow}
\usepackage[]{color}
\usepackage{bbding}
\usepackage{booktabs}
\usepackage{graphicx}
\usepackage{mathtools}
\usepackage{comment}

\usetikzlibrary{patterns}
\usepackage{subcaption}
\usepackage{float}

\makeatletter
\def\resetMathstrut@{%
  \setbox\z@\hbox{%
    \mathchardef\@tempa\mathcode`\[\relax
    \def\@tempb##1"##2##3{\the\textfont"##3\char"}%
    \expandafter\@tempb\meaning\@tempa \relax
  }%
  \ht\Mathstrutbox@\ht\z@ \dp\Mathstrutbox@\dp\z@}
\makeatother
\begingroup
  \catcode`(\active \xdef({\mleft\string(}
  \catcode`)\active \xdef){\mright\string)}
\endgroup
\mathcode`(="8000 \mathcode`)="8000

\newtheorem{theorem}{Theorem}

\newtheorem{observation}{Observation}
\newtheorem{definition}{Definition}
\newtheorem{lemma}{Lemma}
\newtheorem{corollary}{Corollary}

\newcommand{\NN}{\mathbb{N}}
\newcommand{\RR}{\mathbb{R}}
\newcommand{\calL}{\mathcal{L}}
\newcommand{\SC}{\text{SC}}
\newcommand{\instance}{\mathcal{E}}
\newcommand{\dist}{\mathrm{dist}}
\newcommand{\ve}[1]{\mathbf{#1}}

\newcommand{\opt}{\mathit{OPT}}
\newcommand{\optcandidate}{o}

\newcommand{\midd}{\mathfrak{m}}

\newcommand{\xbc}{\midpoint{bc}}
\newcommand{\xab}{\midpoint{ab}}
\newcommand{\xcbtwo}{\midpoint{cb''}}
\newcommand{\midpoint}[1]{\overline{m}_{#1}}

\newcommand{\majorder}{Majority order\xspace}
\newcommand{\focalpoint}{Focal point\xspace}
\newcommand{\step}{\varepsilon}
\newcommand{\dd}{\tau}
\newcommand{\voters}{V} 
\newcommand{\candidates}{C}

\newcommand\myequiv{\stackrel{\mathclap{\small\mbox{3}}}{\equiv}}

\newcommand\myequivv{\stackrel{\mathclap{\small\mbox{2}}}{\equiv}}

\newenvironment{sproof}{%
  \proof}{\endproof}

  \excludecomment{sproof}

\newenvironment{nscenter}
 {\parskip=0pt\par\nopagebreak\centering}
 {\par\noindent\ignorespacesafterend}

\newenvironment{proofcase}
 {
 }

\newenvironment{sfigure}{\begin{figure}}{\end{figure}}

  \excludecomment{sfigure}

\title{Distortion of Multi-Winner Elections on the Line Metric:\\The Polar Comparison Rule}

\author[1]{Negar Babashah \thanks{\href{mailto:negar.babashah@ista.ac.at}{negar.babashah@ista.ac.at}}}
\author[2]{Hasti Karimi \thanks{\href{mailto:hastikm@cs.ubc.ca}{hastikm@cs.ubc.ca}}}
\author[3]{Masoud Seddighin \thanks{\href{mailto:m.seddighin@teias.institute}{m.seddighin@teias.institute}}}
\author[4]{Golnoosh Shahkarami \thanks{\href{mailto:gshahkar@mpi-inf.mpg.de}{gshahkar@mpi-inf.mpg.de}}}

\affil[1]{Institute of Science and Technology Austria (ISTA)}
\affil[2]{University of British Columbia\thanks{This work was carried out during an internship at the Max Planck Institut für Informatik.}}
\affil[3]{Tehran Institute for Advanced Studies (TeIAS), Khatam University}
\affil[4]{Max Planck Institut für Informatik, Universität des Saarlandes}

\date{}
\begin{document}

\maketitle

\begin{abstract}

We study the problem of minimizing metric distortion in multi-winner elections, where a committee of size~$k$ is selected from a set of candidates based on voters' ordinal preferences. We assume that voters and candidates are embedded on a line metric, and social cost is determined by the underlying metric distances.

The \emph{distortion} of a voting rule is the worst-case ratio between the social cost of the elected committee and an optimal committee. Previous work has focused on the \emph{$q$-cost} model, in which a voter’s cost is given by the distance to their $q$th closest committee member. Here, we study the \emph{additive cost}, where a voter’s cost is the sum of distances to all committee members.

We introduce the \emph{Polar Comparison Rule} and analyze its distortion under utilitarian additive cost. We show that it achieves a distortion of at most $2.33$ for all committee sizes $k>2$, improving upon the previously best-known upper bound of~$3$. Moreover, for $k=2$ and $k=3$, we establish tight distortion bounds of $2.41$ and $2.33$, respectively. We also derive lower bounds that depend on the parity of~$k$ and analyze the behavior of distortion for small and large committee sizes. Finally, we extend our results to the egalitarian additive cost.

\end{abstract}

\newpage
\section{Introduction}
Imagine a city council election where residents must select a committee of three representatives from nine candidates. Voters have cardinal costs for each candidate but submit ordinal rankings based on these costs. The problem is how to use these rankings to form a reasonable committee.

To make sense of the above scenario, we need to address two key questions. First, what defines a ``reasonable'' committee? Second, given the limitations of ordinal rankings, can we achieve—or approximately achieve—this objective without access to cardinal values?

What counts as a ``reasonable'' outcome depends on what we expect from the committee and how we balance efficiency and fairness. To explore this, let us set aside the limitations of ordinal rankings for now and look at the example in Figure~\ref{fig:ub(intro)}. The goal is to select a committee of three candidates, where each voter's cost for a candidate is simply their distance from that candidate. If we aim to minimize the \emph{utilitarian social cost}—that is, the total cost across all voters for all committee members—then choosing $b_1, b_2, b_3$ gives the best result. But this outcome might seem unfair. An alternative is to select $a_1, b_1, c_1$, which increases the total cost slightly but ensures that each voter's favorite candidate is included in the committee.

Several well-studied criteria for multi-winner voting rules capture these trade-offs by specifying different objectives based on individual voter costs. These include minimizing the distance from each voter to their nearest committee member \citep{Chen_Li_Wang_2020, cheng2020group}, the distance to their $q$th nearest member \citep{caragiannis2017subset}, or the utilitarian social cost of the committee \citep{10.1145/3230654.3230658}. In this paper, we focus on two classic notions from welfare economics: the utilitarian social cost and the egalitarian social cost, defined as the maximum total cost any individual voter incurs for the entire committee.

The second challenge is the lack of cardinal information. If cardinal values were available, finding the optimal outcome under any of the above criteria would be straightforward. However, assuming access to such detailed information is often unrealistic. Voters may find it difficult to precisely quantify their preferences, and requiring cardinal inputs may lead to inconsistencies, reluctance, or increased costs in data collection. For these and other reasons, ordinal rankings are preferred, despite offering less detailed information about preferences.

The efficiency gap due to the lack of cardinal information is captured by the term \emph{distortion} \citep{procaccia2006distortion}. In single-winner elections, the distortion of a voting rule $f$ is defined as the worst-case ratio (across all instances) between the social cost of the candidate selected by $f$ and that of the optimal candidate, which is based on hidden cardinal values.\footnote{We refer to Section~\ref{sec:prelim} for a formal definition of distortion.}

Over the past decade, distortion has been the subject of extensive study. In general, without any assumptions on voter costs, there exist instances where distortion can be large. As a result, much of the literature has focused on introducing structure or constraints on the cost functions to ensure bounded distortion. Among these, the metric framework has received particular attention. This setting is not only mathematically well-behaved—leading to many elegant and positive results—but also captures a wide range of real-world scenarios, such as spatial models of voting, recommendation systems, and facility location.

For single-winner elections, metric distortion is very well studied. The lower bound for any deterministic rule is~$3$ \citep{ANSHELEVICH201827}, and rules such as Plurality Veto \citep{kizilkaya2022plurality} and Plurality Matching \citep{gkatzelis2020resolving} achieve this bound. Tight bounds have also been established for many other voting rules \citep{ANSHELEVICH201827}.

\begin{figure}[t]
\centering
\small
\begin{tikzpicture}[line cap=round, line join=round, >=triangle 45, x=1cm, y=1cm]
    \draw (3,0.9) node[anchor=center] {$n/3$ voters};
    \draw (0,0.9) node[anchor=center] {$n/3$ voters};
    \draw (-3,0.9) node[anchor=center] {$n/3$ voters};
    
    \draw [line width=2pt] (-3,0)--(3,0);
    
    \draw [fill=black] (-3,0) circle (2pt);
    \draw[color=black] (-3,-0.4) node {$a_1, a_2, a_3$};
    
    \draw [fill=black] (0,0) circle (2pt);
    \draw[color=black] (0,-0.4) node {$b_1, b_2, b_3$};
    
    \draw [fill=black] (3,0) circle (2pt);
    \draw[color=black] (3,-0.4) node {$c_1, c_2, c_3$};
    
    \draw[->] (-3,0.7) -- (-3,0.1);
    \draw[->] (0,0.7) -- (0,0.1);
    \draw[->] (3,0.7) -- (3,0.1);
\end{tikzpicture}
\caption{
An instance with $n$ voters and $9$ candidates. The optimal committee of size $k = 3$ depends on how voter cost is defined in terms of the committee members.
}
\label{fig:ub(intro)}
\end{figure}


In this paper, we examine the distortion of selecting a committee of $k$ candidates under the social additive cost objective. 
The additive cost objective is particularly relevant when all committee members influence the outcome of the committee equally. For example, in a conference, if the reviewers' combined ratings determine the final score of a paper, selecting reviewers whose expertise aligns closely with the paper minimizes the total distance between the paper's field and the reviewers' expertise. 
Note that committee selection with an additive cost objective is equivalent to a randomized voting rule that applies a uniform distribution over the selected candidates with a fixed support size. Thus, the outcome of this rule can also be interpreted as a randomized rule. 

In recent years, several studies have focused on distortion in multi-winner elections. For the utilitarian additive cost objective in the metric framework, it has been shown by \citet{10.1145/3230654.3230658} that repeatedly applying any single-winner rule with distortion~$\delta$ results in a distortion of at most~$\delta$. For instance, applying the Plurality Veto rule \citep{kizilkaya2022plurality} $k$ times yields a committee of size~$k$ with distortion at most~$3$. Additionally, for the $q$-cost objective, \citet{caragiannis2022metric} show that for $q \le k/3$, distortion is unbounded; for $k/3 < q \le k/2$, it is tightly bounded by~$n$; and for $q \ge k/2$, it is tightly bounded by~$3$. However, for the case $q \ge k/2$, they provide a voting rule with distortion~$3$ and exponential running time, and another rule with distortion~$9$ and polynomial running time. \citet{kizilkaya2022plurality} resolve this gap by providing a polynomial-time voting rule with distortion~$3$ for $q \ge k/2$ using Plurality Veto.

We focus on the setting where both voters and candidates are positioned on a line metric. The line metric is an important special case of a metric space, which has been widely studied \citep{black48grouDecision, moulin1980strategy, miyagawa2001locating, fotakis2016conference, fotakis2022distortion, VOUDOURIS2023266, ghodsi2019distortion}. For example, the line metric can model political preferences along a spectrum ranging from liberal to conservative. In this setting, we improve the classic distortion bound of~$3$ for the utilitarian additive cost to approximately~$2.33$.

\subsection{Our Contributions and Techniques}\label{resatec}

The problem of selecting a committee of $k$ candidates based on voters' ordinal preferences is both classical and practically significant. However, choosing more than one candidate introduces new challenges in minimizing the social cost. For minimizing the utilitarian additive cost, choosing consecutive candidates leads to a committee with lower distortion compared to other configurations. Although the median voter seems to be a good representative—and indeed, if selecting a single candidate, the nearest candidate to either the left or right of the median voter would be optimal—returning the entire committee based solely on the preference list of the median voter cannot achieve a distortion better than~$3$. This is because the median voter might prefer all candidates on one side of their location to all candidates on the other side. Therefore, it is important to choose candidates from both sides of the median voter.

To leverage this insight, we design a new voting rule, the \emph{Polar Comparison Rule}, which prioritizes committees that include candidates from both sides of the median voter, thereby increasing the likelihood of selecting better candidates. We first introduce this voting rule for selecting a committee of size two. Specifically, the rule selects the top-ranked candidate of the median voter and then compares the two closest candidates to the median voter on opposite sides of each other that have not yet been picked. Based on the ratio of voters who prefer one candidate to the other, the rule selects the candidate that is more preferred, while introducing a bias toward the candidate on the opposite side of the previously chosen candidate. We then extend this voting rule to select a committee of size three and generalize it to any committee size by iteratively applying the rules for $k = 3$ and $k = 2$.

The \emph{Polar Comparison Rule} achieves upper bounds of $1 + \sqrt{2} \approx 2.41$ and $7/3 \approx 2.33$ distortion for $k = 2$ and $k = 3$, respectively, and we show that these bounds are tight. We further extend these results by showing how different voting rules can be effectively combined to improve the distortion bounds for various committee sizes. Specifically, we generalize our rule to maintain a distortion of $7/3$ for committee sizes divisible by three. For committee sizes where $k \equiv 1 \pmod{3}$, the distortion becomes $7/3 + 4(\sqrt{2} - 4/3)/k$, and for committee sizes where $k \equiv 2 \pmod{3}$, the distortion becomes $7/3 + 2(\sqrt{2} - 4/3)/k$.

Additionally, we complement these results with lower bounds. For each lower bound, we provide two instances where voters share the same preference profile and demonstrate a constraint that any voting rule must satisfy. We establish a lower bound of $2 + 1/k$ for odd values of $k$ and $1 + \sqrt{1 + 2/k}$ for even values of $k$, when $k < m/2$. Furthermore, for larger committee sizes where $k \ge m/2$, we derive a lower bound of $1 + (m - k)/(3k - m)$.

\begin{table}[t]
\centering
    \renewcommand{\arraystretch}{2}
    \begin{tabular}{c c c c}
    \toprule
& \(k = 2\) & \multicolumn{2}{c}{$k > 2$} \\
        \toprule
        \textbf{Upper Bound}
        & $2.41$ 
        [T. \ref{thm:line-voting-rule}]
        & \multicolumn{2}{c}{ 
        $
            \arraycolsep=1.4pt
            \begin{array}{lllr}
            & \begin{array}{cc}
             2.33  \quad &\bigl(k \myequiv 0 \bigr) \\
            2.33 + \frac{4(\sqrt{2} - 4/3)}{k} &\bigl(k \myequiv 1\bigr)\\
            2.33 + \frac{2(\sqrt{2} - 4/3)}{k} &\bigl(k \myequiv 2\bigr)
            \end{array}
            & [\text{T. \ref{thm:ub-general}}]
            \end{array}
        $
        }
        \\
        \midrule
        & \(k = 2\) & \(2 < k < \frac{m}{2}\) & \(k \geq \frac{m}{2}\) \\
        \toprule
        \textbf{Lower Bound}
        &
        $\ 2.41 \ $[T. \ref{thm:lb2.41}]
        &
        $
            \arraycolsep=1.4pt
            \begin{array}{lr}
                \begin{array}{cc}
                1 + \sqrt{1 + \frac{2}{k}} &\bigl(k \myequivv 0\bigr) \\
                2 + \frac{1}{k} &\bigl(k \myequivv 1\bigr)
                \end{array}
            & [\text{T. \ref{thm:lb-general}}]
            \end{array}
        $
        &
        \( 1 + \dfrac{m - k}{3k - m} \)
         [T. \ref{thm:lb->m/2}] \\
        \bottomrule
    \end{tabular}
    \vspace{0.2cm}
\caption{Summary of bounds for the \emph{Polar Comparison Rule}. Note that for $k = 3$, the bounds match $2.33$.
Moreover, an improved upper bound of $2.41$ for $k = 4$ is derived in Theorem~\ref{thm:ub-general}.
}

\label{tab:results}
\end{table}

To establish our results, we rely on several techniques that provide valuable insights beyond the specific context of this problem. One useful technique involves determining the order of candidates on the line metric after removing all candidates that are Pareto-dominated.

Another key observation arises from analyzing how voter movements affect distortion. Specifically, we note that in any instance, moving a single voter to either the right or the left results in a new instance with worse distortion. While this observation does not directly aid our proofs, it prompts us to explore movements that yield instances with guaranteed properties relative to the distortion of the initial instance. We show that there exists a specific point on the line such that moving all voters toward this point results in a new instance where the distortion does not significantly decrease.

To apply this technique effectively, we define certain obstacles that restrict how far voters can move toward this point. Each voter moves until encountering an obstacle, at which point they stop. This approach allows us to precisely determine the structure of the resulting instance. Importantly, the placement of these obstacles depends on the voting rule being analyzed. For our \emph{Polar Comparison Rule}, we position these obstacles at the midpoint between the two candidates being compared. This ensures that the resulting instance aligns with the properties of the rule and facilitates accurate analysis of distortion.

Moreover, we study the egalitarian additive cost and prove that any committee selecting candidates between the two extreme voters achieves a distortion of~$2$. We first show that the maximum cost is incurred by either the leftmost or rightmost voter and then bound both the cost of our selected committee and the optimal one in terms of the distance between these two extreme voters.

\paragraph{Organization of the Paper.}
The remainder of this paper is organized as follows. Section~\ref{sec:prelim} introduces the preliminaries and key definitions. In Section~\ref{sec:order}, we explain how to determine the exact order of candidates on the line and present properties of the optimal candidates. Section~\ref{sec:move-voters} discusses the process of moving voters to construct worst-case instances. In Section~\ref{sec:2-winner}, we provide a warm-up by explaining the high-level idea of our voting rule for $2$-winner elections. Section~\ref{sec:multi-winner} extends these results to general values of~$k$, detailing how to combine different voting rules, extending our \emph{Polar Comparison Rule} for $k = 3$, and establishing an upper bound of roughly~$7/3$, along with lower bounds for general~$k$. In Section~\ref{sec:egalitarian}, we analyze the egalitarian additive cost. Finally, Section~\ref{sec:conclusion} discusses extensions of our rule and outlines open questions for future research. Note that Sections~\ref{sec:order}–\ref{sec:multi-winner} focus on the utilitarian additive cost, while Section~\ref{sec:egalitarian} examines the egalitarian additive cost.

\subsection{Related Work}
As mentioned, distortion was first introduced by \citet{procaccia2006distortion} in the non-metric framework. Since then, several lower and upper bounds have been established for distortion across different voting rules, in various scenarios, and within both non-metric and metric frameworks. For a comprehensive overview of distortion results, we refer to \citet{anshelevich2021distortion}. In this section, we group the related studies into two main areas.

\paragraph{Single-Winner Voting.}
For the metric framework, \citet{ANSHELEVICH201827} established a general lower bound of~$3$ on the distortion of any voting rule. They also proved upper bounds on the distortion of several voting rules, including Majority, Borda, and Copeland, with the best rule being Copeland, which has a distortion of~$5$. Later, \citet{munagala2019improved} generalized this result to uncovered set rules and improved the best upper bound to~$4.236$, and \citet{gkatzelis2020resolving} further reduced it to~$3$. Randomized single-winner voting rules have also been well studied in the metric framework \citep{charikar2022metric, pulyassary2021randomized, fain2019random}. The current best-known upper bound on the distortion of a randomized voting rule, achieved by \citet{charikar2024breaking}, uses a randomization over other voting rules that can achieve a distortion of at most~$2.753$. In the non-metric framework, \citet{caragiannis2011voting} show that any voting rule has a distortion lower bound of $\Omega(m^2)$, with simple rules like Plurality having a distortion of $O(m^2)$.

\paragraph{Multi-Winner Voting.}
In the study of metric distortion in multi-winner voting, various objective functions have been proposed to capture the value each voter derives from the elected committee \citep{elkind2017a, faliszewski2017multiwinner}. \citet{10.1145/3230654.3230658} show that for the additive cost function, applying a single-winner voting rule with distortion~$\delta$ iteratively $k$ times results in a $k$-winner committee with the same distortion~$\delta$. \citet{Chen_Li_Wang_2020} studied the min-cost objective in the \textsc{scv} scenario, where each voter casts a vote for a candidate. In the min-cost objective, each voter's cost is determined by the nearest elected candidate. They propose a voting rule with a tight distortion of~$3$ and a randomized rule with a distortion of $3 - 2/m$. Another objective, introduced by \citet{caragiannis2022metric}, is the $q$-cost objective, where a voter's cost for a committee is defined as the distance to their $q$th closest member in the committee. In the metric setting, \citet{caragiannis2022metric} show that for $q \le k/3$ distortion is unbounded, while for $k/3 < q \le k/2$, distortion is bounded by~$n$. For the case $q > k/2$, they present a non-polynomial voting rule that achieves the tight distortion of~$3$. Subsequently, \citet{kizilkaya2022plurality} propose a polynomial-time voting rule with the same distortion. The study of metric multi-winner voting rules has also been a recent area of focus in several other settings \citep{kalayci2024proportional, pulyassary2022algorithm, Anshelevich21KnownLocations}.

A more specific setting in the metric framework considers single-peaked and $1$-Euclidean preferences, where the voters and candidates lie on the real line \citep{black48grouDecision, moulin1980strategy, miyagawa2001locating, fotakis2016conference, fotakis2022distortion, VOUDOURIS2023266, ghodsi2019distortion}. \citet{fotakis2025distortion} investigate the distortion of deterministic algorithms for $k$-committee selection on the line with the min-cost objective function, using additional distance queries. Recently, \citet{cembrano2025metricdistortionpeerselection} studied this problem in the context of peer selection, where the candidates and voters coincide, and provided tight results for several social cost objectives on the line metric.

For the non-metric framework, and when a voter's utility for a committee is defined as her maximum utility for its members, \citet{caragiannis2017subset} proposed a rule that achieves distortion for deterministic committee selection of $O(1 + m \cdot (m - k)/k)$. Another concept related to committee selection is that of stable committees and stable lotteries, which has been the subject of recent studies \citep{jiang2020approximately, DBLP:conf/ijcai/Borodin0L022}.

\section{Preliminaries}
\label{sec:prelim}
\paragraph{Election.}
An instance of a committee election is a quadruple $\instance = (\voters, \candidates, k, \succ)$, where:
\begin{itemize}
    \item $\voters = [n]$ is the set of voters,
    \item $\candidates = [m]$ is the set of candidates,
    \item $k \in \NN$ is the size of the committee, and
    \item $\ve{\succ} = (\succ_1, \succ_2, \ldots, \succ_n) \in \calL^n(\candidates)$ comprises the voters’ preference profiles, where $\succ_i \in \calL(\candidates)$ is a linear order on $\candidates$ for each $i \in \voters$.
\end{itemize}
For any pair of candidates $a,b \in \candidates$, we let $V_{a \succ b} = \{i \in \voters \mid a \succ_i b\}$ denote the set of voters preferring $a$ over $b$.  
A candidate $a$ \emph{defeats} $b$ if $|V_{a \succ b}| > n/2$, and \emph{Pareto-dominates} $b$ if $|V_{a \succ b}| = n$.

\paragraph{Line metric.}
We assume that all voters and candidates are located on a line metric. Each voter $i \in \voters$ and each candidate $a \in \candidates$ is associated with a position $x_i, x_a \in (-\infty, \infty)$, and the distance between them under metric $d$ is given by $d(i, a) = |x_i - x_a|$.  
We assume that the candidates occupy distinct locations.  
We refer to the (unique) median voter as $\midd$; in the case of an even number of voters, the two consecutive median voters are denoted by $\midd_l$ and $\midd_r$.  
For any interval $I \subseteq \mathbb{R}$, we let $\voters(I) = \{i \in \voters \mid x_i \in I\}$ denote the voters located in $I$; for a single point $x$, we write $\voters(x)$ for the voters located at $\bar{x}$.

A metric $d$ is said to be \emph{consistent} with a preference profile $\succ \in \calL^n(m)$, denoted by $d \rhd \succ$, if for every voter $i \in \voters$ and candidates $a, b \in \candidates$, we have $d(i, a) < d(i, b)$ whenever $a \succ_i b$.  
As before, we may equivalently refer to the vector of locations $x = (x_i)_{i \in \voters}$ and write $x \rhd \succ$.  
A given preference profile may admit multiple consistent embeddings along the line. 

\paragraph{Social cost.}
In our setting, the cost of a committee for a voter is defined as the sum of the distances from that voter to the candidates in the committee (additive cost). The cost of a committee for the entire set of voters can then be defined in two ways: either as the sum of the individual costs across all voters (utilitarian), or as the maximum cost incurred by any voter (egalitarian).

Formally, for a fixed election $\instance = (\voters, \candidates, k, \succ)$, a metric $d$, and a committee $S \subseteq \candidates$ of size $k$, the cost of $S$ for voter $i \in \voters$ is defined using a \emph{candidate-aggregation function} $h \colon \RR^k_+ \to \RR_+$ as
\[
\SC(S, i; d) = h\big((d(i, a))_{a \in S}\big).
\]
In this paper, we focus on the \emph{additive cost}, given by $h(y) = \sum_{a \in S} y_a$.  
Accordingly, the social cost of $S$ for all voters, determined by a \emph{voter-aggregation function} $g \colon \RR^n_+ \to \RR_+$, is
\[
\SC(S, \voters; d) = g\big((\SC(S, i; d))_{i \in \voters}\big).
\]
We study the utilitarian and egalitarian social costs, defined respectively as
\[
g(y) = \sum_{i \in \voters} y_i
\quad \text{and} \quad
g(y) = \max_{i \in \voters} y_i.
\]
Hence, for the \emph{utilitarian additive cost}, we have
\[
\SC (S, \voters; d) = \sum_{i \in \voters} \sum_{a \in S} d(i, a),
\]
and for the \emph{egalitarian additive cost}, we have
\[
\SC(S, \voters; d) = \max_{i \in \voters} \sum_{a \in S} d(i, a).
\]
For a single candidate $a \in \candidates$, its social cost is $\SC(a, \voters; d) = \sum_{i \in \voters} d(i, a)$.  
When $d$ or $\voters$ are clear from the context, we may omit them from the notation; for example, we might write $\SC(S)$, $\SC(a; d)$, or simply $\SC(a)$.
Finally, let $\opt = \{\optcandidate_1, \dots, \optcandidate_k\}$ be an optimal committee of size $k$, i.e., a committee minimizing the utilitarian social cost $\SC(\cdot)$, indexed such that $\SC(\optcandidate_1) \le \dots \le \SC(\optcandidate_k)$.

\paragraph{Voting rules and distortion.}
An $(n, m, k)$-voting rule $f$ takes as input a preference profile $\succ \in \calL^n(m)$ and returns a committee $S$ of size~$k$, where $S = f(\succ) \in {\candidates \choose k}$.  
We assume that $f$ has access only to the rankings, not to the underlying metric.  
For an election $\instance = (\voters, \candidates, k, \succ)$ and a metric $d \rhd \succ$, the distortion of a committee $S \subseteq \candidates$ is defined as
\[
\dist(S, \instance; d)
= \frac{\SC(S, \voters; d)}{\min_{S' \in {\candidates \choose k}} \SC(S', \voters; d)}.
\]
The distortion of $S$ with respect to $\instance$ is its worst-case distortion over all consistent metrics:
\[
\dist(S, \instance) = \sup_{d \rhd \succ} \dist(S, \instance; d).
\]
Finally, the distortion of the rule $f$ is
\[
\dist(f) = \sup_{\succ \in \calL^n(m)} \dist\big(f(\succ), ([n], [m], k, \succ)\big).
\]
Throughout, we study the distortion achievable by voting rules under utilitarian and egalitarian additive social costs.

\section{Properties of Candidates on a Line}
\label{sec:order}
In this section, we first show that the exact order of the candidates on the line can be determined under a mild assumption. We then characterize structural properties of the optimal candidates that later play a crucial role in the selection process. This section serves as a foundation for the analysis and results presented in the subsequent sections.

\subsection{Order of Candidates}
\citet{elkind2014recognizing} show that if the voters’ preference profiles are pairwise distinct, one can uniquely determine the ordering of the voters on the line, up to reversal. They also demonstrate that, given this voter ordering, the candidates can be arranged between the top preference of the leftmost voter and that of the rightmost voter. Although their result was originally used to determine whether an election is $1$-Euclidean, it naturally applies to our setting as well. For completeness, however, we present a simpler method to determine the order of candidates that are not Pareto-dominated by any other candidate.

\begin{restatable}{lemma}{order}
\label{lem:candidates-order}
Under the line metric assumption, if we remove all candidates that are Pareto-dominated by another, the exact order of the remaining candidates on the line can be uniquely determined, up to reversal.
\end{restatable}

To prove Lemma~\ref{lem:candidates-order}, we rely on two simple observations that connect the relative ordering of three candidates to the subsets of voters who prefer one candidate over another.

\begin{observation}\label{obs:subset1}
Given three candidates $a$, $b$, and $c$ with $x_a < x_b < x_c$, the voters who prefer $a$ to $b$ form a subset of those who prefer $a$ to $c$; that is, $V_{a \succ b} \subseteq V_{a \succ c}$. Moreover, $V_{c \succ b} \subseteq V_{b \succ a}$.
\end{observation}

\begin{observation}\label{obs:subset2}
Suppose no candidate Pareto-dominates another, and let $a, c \in \candidates$ be such that $x_a < x_c$. If there exists a candidate $b \in \candidates$ satisfying $V_{a \succ b} \subset V_{a \succ c}$, then $b$ must lie between $a$ and $c$ on the line; that is, $x_a < x_b < x_c$.
\end{observation}

\begin{proof}
If $x_c < x_b$, then by Observation~\ref{obs:subset1} we must have $V_{a \succ c} \subseteq V_{a \succ b}$, which is impossible since $V_{a \succ b} \subset V_{a \succ c}$.
On the other hand, if $x_b < x_a$, Observation~\ref{obs:subset1} implies that $V_{c \succ a} \subseteq V_{a \succ b}$. 
Thus, $V_{c \succ a} \subseteq V_{a \succ b} \subset V_{a \succ c}$. 
But we know that $V_{c \succ a} \cap V_{a \succ c} = \emptyset$, which means $V_{c \succ a} = \emptyset$. 
Hence, candidate~$c$ is dominated by~$a$, contradicting the assumption that all Pareto-dominated candidates have been removed. 
This completes the proof.
\end{proof}

To determine the order of the candidates, we use Algorithm~\ref{alg:order-finder}. 
The algorithm begins by removing all Pareto-dominated candidates. 
It then selects two arbitrary candidates, applies Observations~\ref{obs:subset1} and~\ref{obs:subset2}, and partitions the remaining candidates into left and right subsets. 
Note that if a voter is positioned entirely on one side of all the candidates within a subset, their preference list can be used to establish the relative ordering of those candidates along the line. 
Accordingly, the algorithm identifies two voters whose top preferences lie in the left and right subsets, respectively, and then uses their preference lists to determine the final ordering of candidates.

We are now ready to complete the proof of Lemma~\ref{lem:candidates-order} by showing that Algorithm~\ref{alg:order-finder} indeed recovers the correct order of the candidates on the line.

\begin{nscenter}
\begin{algorithm}[!tb]
    \caption{Finding the order of candidates}
    \label{alg:order-finder}

    \newcommand\mycommfont[1]{\footnotesize\ttfamily\color{blue}#1}%
    \SetCommentSty{mycommfont}

    \textbf{Input:} Election instance $\instance = (\voters, \candidates, k, \succ)$ \\
    \textbf{Output:} Sequence $\mathcal{S}$ of non-Pareto-dominated candidates representing their order on the line

    \medskip

    $B \gets \{\, e \in \candidates \mid \exists\, e' \in \candidates \setminus \{e\} : |V_{e' \succ e}| = n \,\}$ 
    \tcp*{Removing Pareto-dominated candidates}
    $\candidates' \gets \candidates \setminus B$\;
  
    Select $a, b \in \candidates'$ 
    \tcp*{Arbitrary candidates for partitioning}
    $R \gets \{\, c \in \candidates' \mid V_{a \succ b} \subseteq V_{a \succ c} \,\}$\;
    $L \gets \candidates' \setminus R$\;
    $i \gets$ voter whose top-ranked candidate is in $L$\;
    $j \gets$ voter whose top-ranked candidate is in $R$\;
    $\mathcal{S}_L \gets \succ_j(L)$ 
    \tcp*{Order $L$ by voter $j$'s preferences}
    $\mathcal{S}_R \gets \succ_i(R)$ 
    \tcp*{Order $R$ by voter $i$'s preferences}
    $\mathcal{S} \gets \text{reverse}(\mathcal{S}_L) \cdot \mathcal{S}_R$ 
    \tcp*{Reverse $L$’s order and concatenate with $R$}
    \Return $\mathcal{S}$\;
\end{algorithm}
\end{nscenter}

 \begin{proof}[Proof of Lemma~\ref{lem:candidates-order}]
We prove that Algorithm~\ref{alg:order-finder} returns the correct order of the candidates. 
Let $a$ and $b$ be two candidates and, without loss of generality, assume $x_a < x_b$. 
Let $R$ denote the set of all candidates $c$ such that $V_{a \succ b} \subseteq V_{a \succ c}$, and let $L$ be the remaining candidates. 
Note that $a \in L$ and $b \in R$.

\smallskip
\noindent\emph{Claim.} Every candidate in $L$ lies to the left of every candidate in $R$.

\smallskip
\noindent\emph{Proof of claim.}
Suppose for a contradiction that there exist $c \in R$ and $d \in L$ with $x_c < x_d$. 
Since $c \in R$, we have $V_{a \succ b} \subseteq V_{a \succ c}$.

If $V_{a \succ b} \subset V_{a \succ c}$, then by the symmetric form of Observation~\ref{obs:subset2} we must have $x_a \le x_b \le x_c$. 
Hence $x_a \le x_b \le x_d$, and by Observation~\ref{obs:subset1} it follows that $V_{a \succ b} \subseteq V_{a \succ d}$, contradicting $d \notin R$.

It remains to consider the case $V_{a \succ b} = V_{a \succ c}$. 
If $x_c < x_a$, then Observation~\ref{obs:subset1} gives $V_{c \succ a} \subseteq V_{a \succ b}$, and together with $V_{a \succ b} = V_{a \succ c}$ we obtain $V_{c \succ a} \subseteq V_{a \succ c}$. 
But $V_{c \succ a} \cap V_{a \succ c} = \emptyset$, so $V_{c \succ a} = \emptyset$, implying that $a$ Pareto-dominates $c$—a contradiction to our preprocessing step. 
Therefore $x_a \le x_c < x_d$, and Observation~\ref{obs:subset2} yields $V_{a \succ c} \subseteq V_{a \succ d}$. 
Since $V_{a \succ b} \subseteq V_{a \succ c}$, we conclude $V_{a \succ b} \subseteq V_{a \succ d}$, so $d \in R$, again a contradiction. 
This proves the claim. \hfill$\triangle$

\smallskip
We have shown that $L$ and $R$ are nonempty and that all candidates in $L$ lie to the left of all candidates in $R$. 
Consider a voter $i$ whose top-ranked candidate lies in $L$. 
Then $i$ is located to the left of every candidate in $R$, so $i$’s preference list reveals the exact order of the candidates in $R$ along the line. 
Such a voter must exist; otherwise, every voter would prefer the leftmost candidate of $R$ to every candidate in $L$, contradicting our assumption that no candidate Pareto-dominates another. 
A symmetric argument (using a voter $j$ whose top-ranked candidate lies in $R$) provides the order within $L$.

Finally, Algorithm~\ref{alg:order-finder} outputs $\mathcal{S} = \text{reverse}(\mathcal{S}_L) \cdot \mathcal{S}_R$, which places all elements of $L$ before all elements of $R$ and, in particular, lists $a \in L$ before $b \in R$. 
Thus the algorithm recovers the correct order of the non–Pareto-dominated candidates, up to reversal, as claimed.
\end{proof}

\subsection{Optimal Candidate}

We now show that, on either side of the median voter, candidates that are closer to the median voter incur a lower utilitarian social cost.

\begin{lemma}\label{lem:relative-sc}
Let $a, b \in \candidates$ lie on the same side of the median voter (for an even number of voters, consider either $\midd_l$ or $\midd_r$ as the reference). 
If $a$ is closer to the median voter than $b$, then $\SC(a) \le \SC(b)$.
\end{lemma}

\begin{proof}
We consider the case where candidate $a$ lies to the left of the median voter; the other case is symmetric. 
We show that for an odd number of voters, $\SC(a)$ can be expressed as
\[
\SC(a) \;=\; \SC(\midd) \;+\; 2 \sum_{i \in  V([x_\midd, x_a])\setminus \{\midd\}} d(a, i) \;+\; d(\midd, a),
\]
where $\SC(\midd)$ denotes the social cost of a hypothetical candidate located at $\midd$.

For an even number of voters, there are two median voters, $\midd_l$ and $\midd_r$. 
Without loss of generality, suppose $a$ is closer to $\midd_l$ than to $\midd_r$. 
Then $\SC(a)$ can be expressed with respect to $\midd_l$ as follows:
\[
\SC(a) \;=\;
\begin{cases}
\SC(\midd_l), & \text{if } x_{\midd_l} \le x_a \le x_{\midd_r},\\[4pt]
\SC(\midd_l) \;+\; 2 \sum_{i\in V([x_a, x_{\midd_l}])} d(a, i), & \text{if } x_a < x_{\midd_l}.
\end{cases}
\]
Here $\SC(\midd_l)$ and $\SC(\midd_r)$ refer to the social costs of hypothetical candidates located at the two median positions. 
Before proving these equalities, note that Lemma~\ref{lem:relative-sc} follows immediately from the above expressions. 
We now verify each case.

\paragraph{Case 1: The number of voters is odd.}
Partition the rest of the voters into the following subsets:
\begin{align*}
    X &\coloneqq V(\left(x_\midd, +\infty\right)),\\
    Y &\coloneqq V(\left[x_a, x_\midd\right))\\
    Z &\coloneqq V(\left(-\infty, x_a\right])
\end{align*}
By the definition of the median, $|X| = |Y| + |Z|$. Therefore,
\begin{align*}
\SC(a) - \SC(\midd) 
&= d(\midd, a)
  + \sum_{x \in X}\!\big(d(a, x) - d(\midd, x)\big)
  + \sum_{y \in Y}\!\big(d(a, y) - d(\midd, y)\big) + \sum_{z \in Z}\!\big(d(a, z) - d(\midd, z)\big)\\
&= d(\midd, a)
  + |X|\, d(\midd, a)
  + \sum_{y \in Y}\!\big(2 d(a, y) - d(\midd, a)\big)
  + |Z|\, d(\midd, a)\\
&= 2 \sum_{y \in Y} d(a, y) + d(\midd, a)\big(|X| - |Y| - |Z| + 1\big)\\
&= 2 \sum_{y \in Y} d(a, y) + d(\midd, a),
\end{align*}
which is the claimed expression for the odd case.

\paragraph{Case 2: The number of voters is even.}
In this case, $\SC(\midd_l) = \SC(\midd_r)$. 
For any $a$ with $x_{\midd_l} \le x_a \le x_{\midd_r}$ we have $\SC(a) = \SC(\midd_l)$, since moving a candidate within $[x_{\midd_l}, x_{\midd_r}]$ does not change its social cost. 
Hence, assume $x_a < x_{\midd_l}$. 
Partition the voters as follows:
\[
\begin{aligned}
X &\coloneqq V(\left( x_{\midd_r}, +\infty \right)),\\
Y &\coloneqq V(\left( x_{a}, x_{\midd_l} \right)),\\
Z &\coloneqq  V(\left( -\infty, x_{a}\right]).
\end{aligned}
\]
By the definition of the two medians, $|X| = |Y| + |Z|$. Thus,
\begin{align*}
\SC(a) - \SC(\midd_l)
&= \sum_{x \in X}\!\big(d(a, x) - d(\midd_l, x)\big)
 + \sum_{y \in Y}\!\big(d(a, y) - d(\midd_l, y)\big)
 + \sum_{z \in Z}\!\big(d(a, z) - d(\midd_l, z)\big)\\
&\quad + d(a, \midd_l) + \big(d(a, \midd_r) - d(\midd_r, \midd_l)\big)\\
&= 2 \sum_{y \in Y \cup \{\midd_l\}} d(a, y)
  + d(a, \midd_l)\big(|X| - |Y| - |Z|\big)\\
&= 2 \sum_{i\in V([x_a, x_{\midd_l}])} d(a, i),
\end{align*}
which is the claimed expression for the even case. 
This completes the proof.
\end{proof}

\begin{corollary}[of Lemma~\ref{lem:relative-sc}]\label{col:opt-position}
The optimal candidate denoted by $\optcandidate_1$
is either the closest candidate to the right or the closest candidate to the left of the median voter. 
When $n$ is even, the statement applies with respect to both median voters, $\midd_l$ and $\midd_r$.
\end{corollary}

By Lemma~\ref{lem:candidates-order}, we can sort the candidates by their positions on the line. 
Corollary~\ref{col:opt-position} further suggests that it is useful to consider an ordering consistent with the median voter’s preference profile. 
Lemma~\ref{lem:pm-median-ranklist} shows that such an ordering can be obtained. 
(We note that a voter at a given location may admit multiple consistent preference lists due to ties between equidistant candidates.)

\begin{lemma}\label{lem:pm-median-ranklist}
Given an election $\instance = (\voters, \candidates, k, \succ)$, one can obtain an ordering of the candidates that is consistent with the preference list of the median voter.
\end{lemma}

\begin{proof}
Define the \emph{majority graph} of the election $\instance = (\voters, \candidates, k, \succ)$ as the directed graph $\mathcal{G}$ on vertex set $\candidates$ with an arc $(a,b)$ iff $a$ defeats $b$, i.e., $|V_{a \succ b}| > n/2$. 
When $n$ is even and $|V_{a \succ b}| = |V_{b \succ a}|$, we break ties by placing the arc from the candidate that is weakly closer to the other along the line (and hence not to the right of the line); concretely, we orient the tie as $(a,b)$ if $x_a \le x_b$, which is well defined by Lemma~\ref{lem:candidates-order}. 
Under this tie-breaking, $\mathcal{G}$ is a tournament, and thus it admits a Hamiltonian path~\citep{Redei1934}.

We claim that any Hamiltonian path of $\mathcal{G}$ induces an ordering consistent with the median voter’s preference. 
Indeed, for any pair of candidates $a,b \in \candidates$, if the median voter prefers $a$ to $b$ (i.e., $a \succ_\midd b$), then at least half of the voters (those on the same side of the midpoint between $x_a$ and $x_b$ as $\midd$) also prefer $a$ to $b$, so either $a$ defeats $b$ or the pair is tied and the tie is oriented in the direction of $a$ by our rule. 
Hence $a$ precedes $b$ in $\mathcal{G}$’s Hamiltonian path, yielding an ordering that is consistent with the median voter’s preference list.
\end{proof}

In the rest of the paper, we refer to the ordering consistent with the median voter's preference profile, obtained in Lemma~\ref{lem:pm-median-ranklist}, as the \textit{\majorder}.

We now restate a lemma from \citet{ANSHELEVICH201827}, which provides an upper bound on the ratio of social costs between two candidates in any metric setting.

\begin{lemma}[{\citet[Restated]{ANSHELEVICH201827}}]\label{lem:ratio-alters}
For every pair of candidates $a,b \in \candidates$, we have
\[
\frac{\SC(a)}{\SC(b)} \le \frac{2n}{|V_{a \succ b}|} - 1.
\]
\end{lemma}

Finally, we recall Corollary~\ref{col:pm<3opt} from the same work, which follows directly from Lemma~\ref{lem:ratio-alters} and bounds the social cost of any candidate that defeats the optimal one. 
In particular, it implies that the distortion of the leading candidates in the \majorder{} is at most~$3$.

\begin{corollary}[of Lemma~\ref{lem:ratio-alters}]\label{col:pm<3opt}
Let $a \in \candidates$ be a candidate such that $|V_{a \succ \optcandidate_1}| \ge |V_{\optcandidate_1 \succ a}|$. 
Then $\SC(a) \le 3\,\SC(\optcandidate)$.
\end{corollary}

\section{Moving Voters Toward the \focalpoint}
\label{sec:move-voters}
In this section, we present a method to upper-bound the distortion of a committee $S$ compared to an optimal committee $\opt$. We show that it is possible to simplify any election instance by relocating voters, ensuring that if the distortion ratio $\SC(S) / \SC(\opt)$ is initially greater than a threshold $\delta$, it remains greater than $\delta$ in the simplified instance.

The core of this method is a shifting argument. We determine a target location called the \emph{\focalpoint} and shift voters toward it. We shift and collect the voters at specific intermediate locations along their path to the Focal Point. This process results in a simpler instance where voters are located at only a few points on the line. This simplification allows us to easily verify the distortion bound: if we prove that the ratio $\SC(S) / \SC(\opt)$ is smaller than the threshold $\delta$ in the simplified instance, then the ratio must be smaller than $\delta$ in the initial instance as well.

We begin by defining the notion of \emph{consecutive committees} and a simple mathematical statement in Observation~\ref{obs:ratio-increase}.

\begin{definition}
Let $S_1, S_2 \subseteq \candidates$ be committees, and let $T = S_1 \cap S_2$. 
We say that $S_1$ and $S_2$ are \emph{consecutive committees} if the candidates in $S_1 \setminus T$ and $S_2 \setminus T$ lie on opposite sides of $T$. 
If $T = \emptyset$, then $S_1$ and $S_2$ are consecutive if there exists a point on the line such that all candidates in $S_1$ lie strictly on one side of it and all candidates in $S_2$ lie strictly on the other.
\end{definition}

\begin{observation}\label{obs:ratio-increase}
Let $w, z, w', z',$ and $\rho$ be positive real numbers such that $w/z \ge \rho$. Then:
\begin{enumerate}
    \item If $w'/z' \ge \rho$, then $(w + w') / (z + z') \ge \rho$.
    \item If $w'/z' \le \rho$, $w' \le w$, and $z' < z$, then $(w - w') / (z - z') \ge \rho$.
\end{enumerate}
\end{observation}

We are now ready to present Lemma~\ref{lem:move-voters}, which states that given two committees $S_1$ and $S_2$ of size $k$, Algorithm \ref{alg:move-voters-finding-point} returns a point $x^*$ on the line, called the \focalpoint, such that moving all voters toward $x^*$ does not significantly decrease the ratio between the social costs of $S_1$ and $S_2$.

\begin{lemma}\label{lem:move-voters}
Consider an election $\instance = (\voters, \candidates, k, \succ)$, a metric $d$ consistent with $\instance$, and a threshold $\delta > 1$.
Suppose we have two consecutive committees $S_1, S_2 \subseteq \candidates$, each of size~$k$, such that 
\[
\frac{\SC(S_1, \voters; d)}{\SC(S_2, \voters; d)} > \delta.
\]
Let $x^*$ be the \focalpoint on the line returned by Algorithm~\ref{alg:move-voters-finding-point}. 
Consider a new metric $d'$ obtained by moving each voter toward $x^*$. 
Then,
\[
\frac{\SC(S_1, \voters; d')}{\SC(S_2, \voters; d')} > \delta.
\]
\end{lemma}

\begin{nscenter}
\begin{algorithm}[t]
    \caption{Finding the \focalpoint}
    \label{alg:move-voters-finding-point}

    \newcommand\mycommfont[1]{\footnotesize\ttfamily\color{blue}#1}%
    \SetCommentSty{mycommfont}

    \SetKwInOut{Input}{Input}
    \SetKwInOut{Output}{Output}
    
    \Input{Election instance $\instance = (\voters, \candidates, k, \succ)$, 
    metric $d$, consecutive committees $S_1$, $S_2$ $\subseteq \candidates$, and threshold $\delta > 1$.}
    
    \Output{A \focalpoint $x^*$ on the line with respect to $S_1$ and $S_2$.}
    
    \medskip
    
    $T \gets S_1 \cap S_2 = \{\, c_1, c_2, \dots, c_t \,\}$ \tcp*{Common candidates, ordered: $x_{c_1}\le \cdots \le x_{c_t}$}
    
    $L = S_1 \setminus T = \{\, a_1, a_2, \dots, a_{k - t} \,\}$ 
    \tcp*{Left side: $x_{a_1} \le \cdots \le x_{a_{k-t}} \le x_{c_1}$}
    
    $R = S_2 \setminus T = \{\, b_1, b_2, \dots, b_{k - t} \,\}$ 
    \tcp*{Right side: $x_{b_1} \ge \cdots \ge x_{b_{k-t}} \ge x_{c_t}$}

    \eIf{$t \le \lfloor k/2 \rfloor$ or $\delta \le \dfrac{k}{2t -k}$}{
        \Return $x_{b_{\left\lceil \frac{k(\delta - 1)}{2\delta} \right\rceil}}$\;
    }{
    \Return $x_{c_{\left\lceil \frac{k}{2} + \frac{k - t}{\delta - 1} \right\rceil}}$\;    
    }
\end{algorithm}
\end{nscenter}

\begin{proof}
Let $T = S_1 \cap S_2$ with $|T|=t$, $L = S_1 \setminus T$, and $R = S_2 \setminus T$. Based on the definition of consecutive committees and without loss of generalisy, we assume the ordering of candidates on the line as follows:
\[
\underbrace{a_1 \le \dots \le a_{k-t}}_{L} \le \underbrace{c_1 \le \dots \le c_t}_{T} \le \underbrace{b_{k-t} \le \dots \le b_1}_{R}.
\]
We first verify that the algorithm is well-defined. Let $j^* = \lceil \frac{k(\delta - 1)}{2\delta} \rceil$ and $i^* = \lceil \frac{k}{2} + \frac{k - t}{\delta - 1} \rceil$ be the indices used in the algorithm's output. We show that these refer to valid candidates within $R$ (i.e., $1 \le j^* \le k-t$) and $T$ (i.e., $1 \le i^* \le t$), respectively.

\begin{itemize}
    \item If $x^* = b_{j^*}$, we have $j^* = \lceil \frac{k}{2}(1 - \frac{1}{\delta}) \rceil$. Since $\delta > 1$, $j^* \le \lceil k/2 \rceil$.
    If $t \le k/2$, then $|R| = k - t \ge k/2 \ge j^*$, so $b_{j^*}$ exists.
    If $t > k/2$, the algorithm requires $\delta \le \frac{k}{2t - k}$. Substituting this bound yields $j^* \le \lceil \frac{k}{2}(1 - \frac{2t-k}{k}) \rceil = k - t = |R|$. Thus $b_{j^*} \in R$.

    \item If $x^* = c_{i^*}$, this case requires $\delta > \frac{k}{2t - k}$ and $t> \lfloor k/2\rfloor$, which implies $\delta - 1 > \frac{2(k - t)}{2t - k}$. and $2t-k>0$
    Substituting this inequality into the expression for $i^*$:
    \[
    i^* < \frac{k}{2} + \frac{k - t}{2(k - t)/(2t - k)} = \frac{k}{2} + \frac{2t - k}{2} = t.
    \]
    Thus $1 \le i^* \le t$, so $c_{i^*} \in T$.
\end{itemize}
Now consider a voter $v$ moving a distance $\step$ toward the \focalpoint $x^*$ designated by Algorithm~\ref{alg:move-voters-finding-point}, resulting in metric $d'$. We assume that $\step$ is sufficiently small such that $v$ does not pass the location of any candidate in $S_1$ or $S_2$ (though $v$ may land exactly on a candidate). If we prove that this movement keeps the ratio $\SC(S_1, V; d')/\SC(S_2, V; d')$ above $\delta$, we can combine such movements to prove the lemma statement for any other metric caused by moving the voters toward the focal point. 

For any committee $S$, let $\Delta S$ denote the \emph{reduction} in the social cost of $S$ due to this movement.
$$\Delta S = \SC(S, V; d)-\SC(S, V; d').$$
Let $N_{dir}(S)$ be the number of candidates in $S$ located strictly in the direction of movement relative to $v$. The reduction is given by:
\[
\Delta S = (N_{dir}(S) - (k - N_{dir}(S))) \cdot \varepsilon = (2N_{dir}(S) - k) \cdot \varepsilon.
\]
Note that $\Delta S$ can be negative (indicating a cost increase) if $N_{dir}(S) < k/2$.

Let $w = \SC(S_1, \voters; d)$ and $z = \SC(S_2, \voters; d)$. We are given $w/z > \delta$.
The new ratio is $$\dfrac{\SC(S_1, \voters; d')}{\SC(S_2, \voters; d')} = \dfrac{w- \Delta S_1}{z - \Delta S_2}.$$
We verify that this ratio remains greater than $\delta$ by distinguishing the scenarios of cost changes and applying Observation~\ref{obs:ratio-increase}:

\begin{itemize}
    \item Both costs decrease ($\Delta S_1 > 0, \Delta S_2 > 0$): By Observation~\ref{obs:ratio-increase} (Item 2), the ratio remains larger than $\delta$ if the ratio of reductions satisfies $\Delta S_1 / \Delta S_2 \le \delta$.
    \item $S_1$ increases, $S_2$ decreases ($\Delta S_1 \le 0, \Delta S_2 \ge 0$): The numerator increases (or stays same) while the denominator decreases (or stays the same). The ratio increases, so the condition holds trivially.
    \item $S_1$ decreases, $S_2$ increases ($\Delta S_1 \ge 0, \Delta S_2 \le 0$): We must show that this case cannot happen, as the ration will strictly decrease, unless $\Delta S_1 =\Delta S_2 = 0.$
    \item Both costs increase ($\Delta S_1 < 0, \Delta S_2 < 0$): This corresponds to adding positive costs $-\Delta S_1$ and $-\Delta S_2$. By Observation~\ref{obs:ratio-increase} (Item 1), the ratio remains larger than $\delta$ if $-\Delta S_1 / (-\Delta S_2) \ge \delta$.

\end{itemize}

We now analyze the movement based on the location of $v$ relative to $x^*$.

\paragraph{Case A: Moving Right ($v$ is to the left of $x^*$).} 
Let $N_{Right}(S) $ be the number of the candidates of $S$ lying on the right side of $v$. 
 Since $S_2 = T \cup R$ lies to the right of $S_1 = L \cup T$ (with $T$ shared), for any voter position $v$, the number of $S_2$ candidates to the right is always greater than or equal to that of $S_1$, i.e., $N_{Right}(S_2) \ge N_{Right}(S_1)$. Therefore, $\Delta S_2 \ge \Delta S_1$.
\begin{itemize}
    \item If both decrease ($\Delta S_2 > \Delta S_1 > 0$), then $\Delta S_1/\Delta S_2 \le 1 < \delta$, satisfying the condition.
    \item If $\Delta S_1 \le 0$ and $\Delta S_2 \ge 0$, the ratio increases and remains larger than $\delta$.
    \item The case $\Delta S_1 \ge 0$ and $\Delta S_2 \le 0$ is not possible because $\Delta S_2 \ge \Delta S_1$, unless $\Delta S_1 = \Delta S_2 =0$, for which the lemma's statement holds.
    \item If both increase ($\Delta S_1 <\Delta S_2 < 0$), we have $-\Delta S_1 \ge -\Delta S_2 > 0$. We need to prove that $-\Delta S_1 /(-\Delta S_2)  \ge \delta$, or equivalently
    $$\dfrac{k - 2N_{Right}(S_1)}{k - 2N_{Right}(S_2)} \ge \delta \iff 2\left(N_{Right}(S_2)\delta - N_{Right}(S_1)\right) \ge k(\delta -1) .$$
    Now we have three cases for the position of $v$: 
    \begin{itemize}
    \item 
    If $v$ lies on the left side of $c_1$, then $N_{Right}(S_1)\le k$ and $N_{Right}(S_2)=k$. Therefore
    $2(N_{Right}(S_2)\delta - 2N_{Right}(S_1)) \ge k(\delta -1)$, as we wanted.
    \item 
    If $v$ lies on the right side of $b_{k-t}$, the fact that $v$ is to the left of $x^*$ implies that $x^* \in R$. Therefore, according to Algorithm \ref{alg:move-voters-finding-point}, $x^*= x_{b_{j^*}}$ where $j^* = \lceil k(\delta-1)/(2\delta)\rceil$. Thus, $N_{Right}(S_1) = 0$ and $N_{Right}(S_2) \ge j^*\ge k(\delta-1)/(2\delta)$. Therefore,
    \begin{align*}
        2(N_{Right}(S_2)\delta - N_{Right}(S_1)) \ge 2j^*\delta \ge k(\delta-1).
    \end{align*}
    \item 
    Finally, if $v$ lies between $c_1$ and $b_{k-t}$, we have two subcases according to the algorithm: the \focalpoint $x^*$ is $x_{c_{i^*}}$ for $i^*=\lceil k/2 + (k-t)/(\delta - 1)\rceil$, or the output is $x_{b_{j^*}}$ where $j^*=\lceil k(\delta-1)/(2\delta)\rceil$. For both subcases, assume that there are $i$ members of $T$ to the left side of (or located on) $v$. Therefore, $N_{Right}(S_1) = t-i, N_{Right}(S_2) = k-i$.
    
    In the first subcase, since $v$ must lie to the left side of $c_{i^*}$, $i\le i^*-1\le k/2 + (k-t)/(\delta-1)$. This implies that
    \begin{align*}
        2(N_{Right}(S_2)\delta - N_{Right}(S_1)) &= 2(k\delta - i\delta -t + i)\\
        & = 2k\delta - 2t -2i(\delta - 1)\\
        & \ge 2k\delta - 2t -2(\delta - 1)(\frac{k}{2 } + \frac{k-t}{\delta - 1})\\
        & = k(\delta - 1).
    \end{align*}
    For the second subcase, $x^* = x_{b_{j^*}}$. According to the construction of the algorithm, this happens if either $t \le \lfloor k/2\rfloor$ or $\delta\le k/(2t-k)$. Moreover, since $v$ lies on the left side of $b_{k-t}$, $i\le t$. Therefore,
    \begin{align*}
        2(N_{Right}(S_2)\delta - N_{Right}(S_1)) &= 2(k\delta - i\delta -t + i)\\
        & = 2(k\delta -t - i(\delta -1))\\
        &\ge 2(k\delta -t -t(\delta - 1))\\
        & = 2\delta (k-t).
    \end{align*}
    We want to prove that $2\delta (k-t) \ge k\delta - k$. This is equivalent to
    \begin{align*}
        2\delta (k-t) \ge k\delta - k & \iff k (\delta + 1) \ge 2\delta t\\
        & \iff \dfrac{\delta + 1}{\delta}\ge \dfrac{2t}{k}\\
        &\iff \dfrac{1}{\delta}\ge \dfrac{2t-k}{k}.
    \end{align*}
    We know that in this subcase, either $t \le \lfloor k/2\rfloor$ or $\delta\le k/(2t-k)$. If $2t-k\le 0$, the above inequality holds becaue $1/\delta >0\ge (2t-k)/k$. Otherwise, $t> \lceil (k+1)/2\rceil$. Therefore, we must have $\delta \le k/(2t-k)$. Thus, 
    $1/\delta \ge (2t-k)/k$, proving the inequality.
    \end{itemize}

\end{itemize}

\paragraph{Case B: Moving Left ($v$ is to the right of $x^*$).}
Since $S_1 = L \cup T$ lies to the left of $S_2 = T \cup R$, for any voter position $v$, the number of candidates in $S_1$ located strictly to the left of $v$ is greater than or equal to that of $S_2$. Let $N_{Left}(S)$ denote the number of candidates of $S$ strictly to the left of $v$. We have $N_{Left}(S_1) \ge N_{Left}(S_2)$, which implies $\Delta S_1 \ge \Delta S_2$.
First, we show that the social cost of $S_2$ decreases or stays the same (i.e., $\Delta S_2 \ge 0$). This requires $N_{Left}(S_2) \ge k/2$.
Since $v$ is to the right of $x^*$, $N_{Left}(S_2)$ is at least the number of candidates in $S_2$ to the left of (or at) $x^*$.
\begin{itemize}
    \item If $x^* = b_{j^*} \in R$, then $N_{Left}(S_2) \ge t + (k - t - j^* + 1)$. Since $j^* \le \lceil k/2 \rceil$, this sum is at least $k - \lceil k/2\rceil + 1 \ge k/2$.
    \item If $x^* = c_{i^*} \in T$, then $N_{Left}(S_2) \ge i^*$. Since $i^* \ge k/2$, this is at least $k/2$.
\end{itemize}
Thus $\Delta S_2 \ge 0$, and consequently $\Delta S_1 \ge 0$.
We are in the scenario where both costs decrease. The condition to maintain the ratio is $\Delta S_1 / \Delta S_2 \le \delta$, which simplifies to:
\begin{equation}\label{eq:move-voters-caseB}
2(\delta N_{Left}(S_2) - N_{Left}(S_1)) \ge k(\delta - 1).
\end{equation}
We analyze this inequality based on the branches of the algorithm and the position of $v$.

\begin{itemize}
    \item If the algorithm returns $x^* = b_{j^*}$:
    Here $x^* \in R$. Since $v$ is to the right of $x^*$, $v$ must be to the right of all candidates in $T$ and $L$. Therefore, $N_{Left}(S_1) = k$.
    For $S_2$, $v$ is to the right of all of $T$ and at least the candidates $b_{k-t}, \dots, b_{j^*}$ in $R$. The number of candidates in $R$ strictly to the left of $v$ is at least $(k - t) - j^* + 1$. Thus, $N_{Left}(S_2) \ge t + k - t - j^* + 1 = k - j^* + 1$.
    Substituting these into Inequality \eqref{eq:move-voters-caseB}:
    \begin{align*}
        2(\delta N_{Left}(S_2) - N_{Left}(S_1)) &\ge 2(\delta(k - j^* + 1) - k) \\
        &= 2\delta k - 2\delta j^* + 2\delta - 2k.
    \end{align*}
    We need this to be $\ge k\delta - k$. Rearranging terms, we require:
    \[
    k\delta + 2\delta - k \ge 2\delta j^* \iff j^* \le \frac{k(\delta - 1)}{2\delta} + 1.
    \]
    Since the algorithm sets $j^* = \lceil \frac{k(\delta - 1)}{2\delta} \rceil$ and $\lceil x \rceil < x + 1$, this condition always holds.

    \item If the algorithm returns $x^* = c_{i^*}$:
    Here $x^* \in T$. The voter $v$ is to the right of $c_{i^*}$. We distinguish two positions for $v$:
    \begin{itemize}
        \item If $v$ lies within $T$ (specifically, $v$ is between $c_i$ and $c_{i+1}$ for some $i \ge i^*$, or between $c_t$ and $b_{k-t}$):
        Here, the number of $T$ candidates to the left of $v$ is $i$, where $i \ge i^*$.
        We have $N_{Left}(S_1) = (k - t) + i$ and $N_{Left}(S_2) = i$.
        Inequality \eqref{eq:move-voters-caseB} becomes:
        \[
        2(\delta i - (k - t + i)) = 2(i(\delta - 1) - (k - t))\ge k(\delta + 1).
        \]
        Solving this for $i$:
        \[
        2i(\delta - 1) \ge k(\delta - 1) + 2(k - t) \iff i \ge \frac{k}{2} + \frac{k - t}{\delta - 1}.
        \]
        Since $v$ is to the right of $c_{i^*}$, we know $i \ge i^*$. The algorithm sets $i^* = \lceil \frac{k}{2} + \frac{k - t}{\delta - 1} \rceil$, ensuring this condition is satisfied.

        \item If $v$ lies within $R$ (or to the right of $R$):
        Here, $v$ is to the right of all candidates in $T$. Thus $N_{Left}(S_1) = k$.
        For $S_2$, $v$ is to the right of all $T$ and some non-negative number of candidates in $R$, say $n_R \ge 0$. So $N_{Left}(S_2) = t + n_R$.
        Inequality \eqref{eq:move-voters-caseB} becomes:
        \[
        2(\delta(t + n_R) - k) \ge k(\delta -1).
        \]
        Since $n_R \ge 0$, we require:
        \[
        2(\delta t - k) \ge k(\delta - 1) \iff 2\delta t - 2k \ge k\delta - k \iff \delta(2t - k) \ge k.
        \]
        This branch of the algorithm is only entered when $2t - k > 0$ and $\delta > \frac{k}{2t - k}$, which implies $\delta(2t - k) > k$. Thus, the inequality holds.
    \end{itemize}
\end{itemize}
\end{proof}

Using the fact that moving voters toward the \focalpoint keeps the ratio above $\delta$ by Lemma~\ref{lem:move-voters}, the following lemma shows that we can gather these voters at specific points $x_1, \dots, x_t$ and $x^*$. This allows us to replace the original election with a simplified instance where all voters are located at just a few positions.

\begin{lemma}\label{lem:move-voters-with-prop}
Let $\instance=(\voters,\candidates,k,\succ)$ be an election and let $d$ be a consistent metric.
Let $S_1,S_2 \subseteq \candidates$ be consecutive committees of size $k$ with
\[
\frac{\SC(S_2,\voters;d)}{\SC(S_1,\voters;d)}>\delta,
\]
and let $x^*$ be the \focalpoint with respect to threshold $\delta$, $S_1$, and $S_2$.
Suppose we have:
\begin{itemize}
    \item A sequence of points $z_1 \le \dots \le z_t \le x^*$  and integers $0=r_0 \le r_1 \le \dots \le r_t$ such that for each $i \in \{1,\dots,t\}$,
    \[
    \left|\voters(\left(-\infty, z_i\right])\right| \ge r_i.
    \]
    \item A sequence of points $y_1 \ge \dots \ge y_p \ge x^*$ and integers $0=q_0 \le q_1 \le \dots \le q_p$ such that for each $j \in \{1,\dots,p\}$,
    \[
    \left|\voters(\left[y_j, \infty\right))\right| \ge q_j.
    \]
\end{itemize}
Assume that the sets of voters counted are disjoint (i.e., $r_t + q_p \le |\voters|$).
Construct $\instance'=(\voters,\candidates,k,\succ')$ with metric $d'$ by relocating voters as follows:
\begin{enumerate}
    \item For each $i \in \{1,\dots,t\}$, place exactly $r_i - r_{i-1}$ voters at $z_i$.
    \item For each $j \in \{1,\dots,p\}$, place exactly $q_j - q_{j-1}$ voters at $y_j$.
    \item Place all remaining voters at $x^*$.
\end{enumerate}
Then,
\[
\frac{\SC(S_2,\voters;d')}{\SC(S_1,\voters;d')}>\delta.
\]
\end{lemma}

\begin{proof}
We index the voters as $v_1, v_2, \dots, v_n$ such that their positions in the original metric $d$ are sorted non-decreasingly: $x_{v_1} \le x_{v_2} \le \dots \le x_{v_n}$.
We construct the new instance by moving specific groups of voters to the locations $z_i$, $y_j$, or $x^*$. We show that every voter moves toward $x^*$ (or stays fixed), which by Lemma~\ref{lem:move-voters} implies the social cost ratio remains greater than $\delta$.

\textbf{Left Side ($x_v \le x^*$):}
The constraint $|\voters(\left(-\infty, z_i\right])| \ge r_i$ implies that at least $r_i$ voters are located at or to the left of $z_i$. In our sorted indexing, this means the voter $v_{r_i}$ satisfies $x_{v_{r_i}} \le z_i$. Consequently, for any $k \le r_i$, the voter $v_k$ is located at $x_{v_k} \le x_{v_{r_i}} \le z_i$.
For each $i \in \{1, \dots, t\}$, we select the block of voters $\{v_k \mid r_{i-1} < k \le r_i\}$ and move them to $z_i$.
For any such voter $v_k$, the movement is from $x_{v_k}$ to $z_i$. Since $x_{v_k} \le z_i \le x^*$, this is a movement to the right, strictly toward the \focalpoint $x^*$.

\textbf{Right Side ($x_v \ge x^*$):}
The construction for the right side is symmetric. The constraint $\left|\voters(\left[y_j, \infty\right))\right| \ge q_j$ ensures that the last $q_j$ voters are to the right of $y_j$; moving the corresponding block to $y_j$ constitutes a movement to the left, toward $x^*$.

\textbf{Remaining Voters:}
The remaining voters are those with indices $k$ such that $r_t < k \le n - q_p$. We move all these voters to $x^*$.
Regardless of whether a voter $v_k$ is originally to the left or right of $x^*$, moving them directly to $x^*$ reduces their distance to $x^*$ to zero. Thus, they are moved toward the \focalpoint.

Since all voters are moved toward $x^*$ (or remain at their position), the condition of Lemma~\ref{lem:move-voters} is satisfied, and the ratio of social costs in the new instance remains strictly greater than $\delta$.
\end{proof}

\section{Warm-up: Polar Comparison Rule for $k=2$}
\label{sec:2-winner}
In this section, we provide a voting rule for the line metric in the case where the committee size is two (i.e., a $2$-winner election), and we prove that the distortion of this rule is at most $1 + \sqrt{2}$. Moreover, we show that this bound is tight for the line metric. This section serves as a warm-up for the next section, where we generalize the approach to elections with larger committees.

As described in Lemma \ref{lem:relative-sc}, the candidates in the optimal committee are close to the median voter(s). Thus, a natural approach is to select the committee from the candidates near the top of the \majorder, since these candidates are near the median voter(s). In Algorithm \ref{alg:2-winner-voting-rule}, we present a voting rule for $2$-winner elections ($k=2$), called \emph{Polar Comparison Rule}.
This voting rule first selects the top candidate in the \majorder. Let us denote this candidate by $a$. To select the second member, we compare the two candidates adjacent to $a$, denoted by $b$ and $c$, where $b$ is the second candidate in the \majorder and $c$ is the closest candidate to $a$ such that $a$ lies between $b$ and $c$. Without loss of generality, assume that $x_c \leq x_a \leq x_b$. The intuition behind this approach is to ensure that the output committee is favored by voters on both sides of $a$.
First, we prove that we can identify the three candidates $a$, $b$, and $c$ under some mild conditions using the ordinal preferences $\succ$.

\begin{observation}
    Given an election instance $\instance = (\voters, \candidates, 2, \succ)$, if there exists a candidate $c$ immediately to the left of $a$, such that $c$ is not Pareto-dominated by any candidate other than $a$ (and possibly not Pareto-dominated at all), we can identify the three candidates $a$, $b$, and $c$ as described above using only the ordinal preferences $\succ$.
\end{observation}

\begin{proof}
    To identify $a$ and $b$, we apply Lemma \ref{lem:pm-median-ranklist} to establish the \majorder. 
    Next, to determine the candidate $c$ immediately to the left of $a$, we refine the set of eligible candidates. This is because $c$ might be Pareto-dominated by $a$. We then apply Lemma \ref{lem:candidates-order} to determine the exact order of the remaining candidates along the line. From this ordered set, we identify $c$ as the candidate closest to $a$ such that $a$ lies between $b$ and $c$. If no such candidate exists, then either there is no candidate to the left of $a$, or all candidates to the left of $a$ are Pareto-dominated.
\end{proof}

\begin{nscenter}
\begin{algorithm}[t]
    \caption{$2$-Winner Polar Comparison Rule}
    \label{alg:2-winner-voting-rule}

    \newcommand\mycommfont[1]{\footnotesize\ttfamily\color{blue}#1}%
    \SetCommentSty{mycommfont}
    
    \SetKwInOut{Input}{Input}
    \SetKwInOut{Output}{Output}
    
    \Input{Instance $\instance = (\voters, \candidates, 2, \succ)$}
        
    \Output{A committee of two candidates}

    \medskip
    $a, b \gets$ top two candidates in the \majorder\;
    $B \gets \{e \in \candidates \mid \exists e' \in \candidates\setminus \{e, a\} : |V_{e' \succ e}| = n\}$
    \textcolor{blue}{\tcp*{Candidates Pareto-dominated by $e'\ne a$}}
    $\candidates' \gets \candidates \setminus B$
    \textcolor{blue}{\tcp*{Exclude Pareto-dominated candidates}}
    Find the closest candidate $c\in \candidates'$ to $a$ such that $a$ lies between $b$ and $c$
    \textcolor{blue}{\tcp*{By Lemma~\ref{lem:candidates-order}}}
    \If{$c$ does not exist}{
        \Return $\{a, b\}$\;
    }
    \eIf{$|V_{c \succ b}| \le \frac{n}{1 + \sqrt{2}}$}{
      \Return $\{a, b\}$\;
    }{
      \eIf{$|V_{b \succ a}| \ge \frac{n}{1 + \sqrt{2}}$}{
        \Return $\{a, b\}$\;
      }{
        \Return $\{a, c\}$\;
      }
    }
    
\end{algorithm}
\end{nscenter}

We now show that our $2$-Winner Polar Comparison Rule achieves a distortion bounded by $1+\sqrt{2}$ under the utilitarian additive cost on the line metric.

\begin{theorem}\label{thm:line-voting-rule}
Let $f$ be the voting rule that, for any election instance 
$\instance = (\voters, \candidates, k, \succ)$,
selects a committee of size two according to Algorithm~\ref{alg:2-winner-voting-rule} (the $2$-Winner Polar Comparison Rule).  
Assume that voters and candidates lie on a line metric, and let 
$d$ be any metric consistent with~$\succ$.
Then, under the utilitarian additive social cost,
\[
\dist(f(\succ), \instance; d) \le 1+\sqrt{2}.
\]
In particular, the distortion of $f$ satisfies 
\[
\dist(f) \le 1+\sqrt{2}.
\]
\end{theorem}

To prove Theorem~\ref{thm:line-voting-rule}, we use Observation~\ref{obs:set-of-voters-diff-view}, which states that the voters who prefer candidate \(a\) over \(b\) are exactly those whose locations lie on the half-line containing \(a\) and bounded by the midpoint between \(a\) and \(b\), namely $\xab$.

\begin{observation}\label{obs:set-of-voters-diff-view}
    For two candidates $a$ and $b$, if $x_a \leq x_b$, then $V_{a \succ b} = \voters(\left(-\infty, \xab\right])$, where $\xab$ is the midpoint of $x_a$ and $x_b$.  
    Similarly, if $x_a \geq x_b$, then $V_{a \succ b} = \voters(\left[\xab, \infty\right))$.
\end{observation}

As an intermediate step, we prove that if the median voter lies between the top two candidates in the 
\majorder, then selecting these two candidates yields a distortion of at most 2. This lemma will be instrumental in the proof of Theorem~\ref{thm:line-voting-rule}.

\begin{lemma}\label{lem:ub-3-le2}
    Let $\instance$ be an election where $a$ and $b$ are the top two candidates in the \majorder.
    If $\midd$ lies between $a$ and $b$, then for any pair of candidates $a'$ and $b'$,
    \[
        \frac{\SC(\{a, b\})}{\SC(\{a', b'\})} \le 2.
    \]
\end{lemma}

\begin{proof}
    Fix an instance $\instance$. Without loss of generality, assume that $x_{a} \le x_{b}$.
    It suffices to prove the lemma for the case that $\{ a', b'\} = \{\optcandidate_1, \optcandidate_2\} $ is the optimal committee of size 2.
    Note that by Corollary \ref{col:opt-position}, the optimal candidate, $\optcandidate_1$, is either $a$ or $b$. 
    Without loss of generality, assume that $a = \optcandidate_1$. If $b \neq \optcandidate_2$, by Lemma \ref{lem:relative-sc}, $x_{\optcandidate_2} \le x_{a} \le x_{b}$ and also $\SC(a) \ge \SC(\midd)$. Therefore,
    \begin{align}\label{eq:ub-3-le2}
    \begin{split}
        \dfrac{\SC(\{a, b\})}{\SC(\{a, \optcandidate_2\})} &\le \dfrac{\SC(\{\midd, b\})}{\SC(\{\midd, \optcandidate_2\})}\\
        & = 1 + \dfrac{\SC(b) - \SC(\optcandidate_2)}{\SC(\midd) + \SC(\optcandidate_2)}.
    \end{split}
    \end{align}
    Next, we transform instance $\instance$ to a new instance $\instance'$, where the positions of voters are changed by a metric $d'$ in which the ratio in Equation \eqref{eq:ub-3-le2} does not decrease. This new instance will help us to upper bound the ratio in Equation \eqref{eq:ub-3-le2}. The new positions, changed by a metric $d'$ are defined as follows:
    \begin{equation}\label{move_ub-3-le2}
        x_i' =
        \begin{cases}
        x_{\optcandidate_2} & \text{ if } x_i \in \left( -\infty, x_{\midd} \right)\\
        x_{\midd} & \text{ if } x_i \in \left[x_{\midd}, + \infty\right).
        \end{cases}
    \end{equation}
    Note that in comparison to $\instance$, the voters in $\instance'$ are moved to new locations according to Equation \eqref{move_ub-3-le2}. Therefore, there are two types of movements. We show that none of these movements reduce the ratio in Equation \eqref{eq:ub-3-le2}.
    \begin{itemize}
        \item $\mathbf{\left( -\infty, x_{\midd} \right) \rightarrow x_{\optcandidate_2}}$: If at first we move every voter in $(-\infty, x_{\optcandidate_2})$ to $x_{\optcandidate_2}$, this move does not change the value of $\SC(b) - \SC(\optcandidate_2)$, but decreases $\SC(\midd) + \SC(\optcandidate_2)$. 
        Afterwards, if we move all voters in $(x_{\optcandidate_2}, x_{\midd})$ to $x_{\optcandidate_2}$, the new value of $\SC(\midd) + \SC(\optcandidate_2)$ remains intact, while the value of $\SC(b) - \SC(\optcandidate_2)$ is increased. Therefore, this move does not decrease the total ratio in Equation \eqref{eq:ub-3-le2}.
    
        \item $\mathbf{\left[x_{\midd}, +\infty \right) \rightarrow x_{\midd}}$: If at first we move every voter in $(x_{b}, + \infty)$ to $x_{b}$, this move does not change the value of $\SC(b) - \SC(\optcandidate_2)$, but decreases $\SC(\midd) + \SC(\optcandidate_2)$. 
        Afterwards, if we move all voters in $\left(x_{\midd}, x_{b}\right]$ to $x_{\midd}$, the new value of $\SC(\midd) + \SC(\optcandidate_2)$ is decreased, while the value of $\SC(b) + \SC(\optcandidate_2)$ is unchanged. Therefore, this move does not reduce the total ratio in Equation \eqref{eq:ub-3-le2}.
    \end{itemize}
    Now it remains to calculate the ratio in the new instance $\instance'$. Note that since $b$ appears earlier than $\optcandidate_2$ in the \majorder, $|V_{b \succ \optcandidate_2}| \ge n/2$, see Figure \ref{fig:ub-3-le2}. For brevity, let $s = d(\optcandidate_2, \midd)$ and $r = d(\midd, b)$. Note that $s \ge r$. Therefore,
    \begin{align*}
        \dfrac{\SC(\{a, b\})}{\SC(\{a, \optcandidate_2\})} &\le \dfrac{\SC(\{\midd, b\})}{\SC(\{\midd, \optcandidate_2\})}\\
        & \le \dfrac{\frac{n}{2} (s + r + (s + r))}{ns} \\
        & \le 2 .
    \end{align*}
    
    \begin{figure}[t]
        \centering
        \begin{tikzpicture}[line cap=round,line join=round,>=triangle 45,x=1cm,y=1cm]
        
        \draw [line width=2pt,color=black] (-3,0)-- (3,0);

        \draw (0,1) node[anchor=center] {$|\voters(\midd)|\ge \frac{n}{2}$};
        \draw (-3,1) node[anchor=center] {$|\voters(\optcandidate_2)| \le \frac{n}{2}$};
    
        \draw [fill=black] (3,0) circle (1.5pt);
        \draw[color=black] (3,-0.5) node {$b$};
        
        \draw [fill=black] (-3,0) circle (1.5pt);
        \draw[color=black] (-3,- 0.45) node {$\optcandidate_2$};
        
        \draw [fill=black] (0,0) circle (1.5pt);
        \draw[color=black] (0,-0.5) node {$\midd$};

         \draw [fill=black] (-1.7,0) circle (1.5pt);
        \draw[color=black] (-1.7,-0.5) node {$\optcandidate_1 = a$};

        \draw[->]  (0, 0.7) to (0, 0.1);
        \draw[->]  (-3, 0.7) to (-3, 0.1);
        
        \end{tikzpicture}
        \caption{
        If $a$ and $b$ are the top two candidates in the \majorder and $\{\optcandidate_1, \optcandidate_2\}$ is the optimal committee of size 2, by assuming that $a = \optcandidate_1$ and moving the location of voters, one can prove that 
        $\dist(\{a, b\}) \le 2$.
        }
        \label{fig:ub-3-le2}
    \end{figure}

\end{proof}

\begin{sproof}[Sketch of Proof of Theorem \ref{thm:line-voting-rule}]
We analyze different cases of the optimal committee and the output of \(f\). In some cases, we apply Lemmas \ref{lem:ratio-alters} and \ref{lem:ub-3-le2} to establish an upper bound on the distortion. In other cases, we identify the \focalpoint and move voters according to Lemma \ref{lem:move-voters-with-prop} to construct a new instance, for which the social cost ratio does not exceed \(1 + \sqrt{2}\).
\end{sproof}

Now, we present the proof of Theorem \ref{thm:line-voting-rule}, which establishes the upper bound of \(1 + \sqrt{2}\) on the distortion of the $2$-Winner Polar Comparison Rule.
\begin{proof}[Proof of Theorem \ref{thm:line-voting-rule}]
To prove the theorem, we begin by assuming the contrary, i.e., there exists an instance $\instance = (\voters, \candidates, 2, \succ)$ and metric $d$, such that $\dist(f(\succ), \instance; d) > 1 + \sqrt{2}$. Assume the top two candidates in the \majorder are $a$ and $b$, and without loss of generality, $x_a \leq x_b$. Let $c$ be the nearest candidate to $a$ that is not Pareto-dominated by any other candidate and satisfies $x_c \leq x_a$.

If such a candidate $c$ does not exist, it implies that $a$ and $b$ Pareto-dominate all candidates to the left of $a$. Therefore, by Lemma \ref{lem:pm-median-ranklist}, either $\{a, b\}$ is the optimal committee, or $b$ is the optimal candidate. In the latter case, by Corollary \ref{col:pm<3opt}, $\SC(a) \leq 3 \SC(b)$. Thus,

\[
\dist(f(\succ), \instance; d) = \frac{\SC(a) + \SC(b)}{2\SC(b)} \leq 2 < 1 + \sqrt{2}.
\]
Therefore, we assume that such a candidate $c$ exists. By Corollary \ref{col:opt-position}, we observe that $\optcandidate_1 \in \{a, b, c\}$.
Furthermore, by Lemma \ref{lem:relative-sc}, either $\optcandidate_2 \in \{a, b, c\}$, or we can lower bound $\SC(\opt)$ by one of $2\SC(b)$ and $2\SC(c)$. Thus, we can consider the following cases for the optimal candidates:
    \[ (x_{\optcandidate_1}, x_{\optcandidate_2}) \in \{ (x_b, x_b), (x_c, x_c), (x_a, x_c), (x_c, x_a), (x_a, x_b), (x_b, x_a)\}. \]

    We evaluate the distortion of the voting rule on a case-by-case basis and establish that it is always at most \(1+\sqrt{2}\). Throughout the proof, we denote by \(\xbc\) the midpoint of \(x_b\) and \(x_c\), and by \(\xab\) the midpoint of \(x_a\) and \(x_b\). We furthermore introduce a parameter \(\alpha\), which we fix to \(\sqrt{2}\), since this choice minimizes the upper bound on the distortion in every case.

    \begin{enumerate}[label=\textbf{Case.\arabic*},ref=Case \arabic*, align=left]
    
\item \label{case:1} $(x_{\optcandidate_1}, x_{\optcandidate_2}) = (x_b,x_b)$.

\begin{proofcase}
    Our voting rule, $f$, selects $a$ and then chooses between $c$ and $b$.
    Since $b$ is the optimal candidate, by Corollary \ref{col:opt-position}, the median point, $\midd$, must lie between $a$ and $b$. Therefore, for the case that our voting rule selects $a$ and $b$, by Lemma \ref{lem:ub-3-le2}, we have:
    \begin{equation*}
    \dist(f(\succ), \instance; d) \leq 2 < 1 + \sqrt{2}.
    \end{equation*}
    Now consider the case that $f$ has chosen $a$ and $c$. By the construction of Algorithm \ref{alg:2-winner-voting-rule}, this happens when $|V_{c\succ b}| > n/(1 +\alpha)$ and $|V_{b \succ a}| < n/(1 + \alpha)$. Therefore, by Observation \ref{obs:set-of-voters-diff-view}, we have that $\left|\voters((-\infty, \xbc))\right| > n/(1 +\alpha)$, and $\left|\voters((-\infty, \xab))\right| \geq n\alpha/(1 + \alpha)$. The distortion of $f$ for instance $\instance$ in this case is:
    \[\dist(f(\succ), \instance; d ) = \frac{\SC(a) + \SC(c)}{2\SC(b)}.\]
    Note that for this case, based on Algorithm \ref{alg:move-voters-finding-point}, the \focalpoint would be the point $x_b$.
     Here, again, we consider two cases. The first is when $x_a$ resides between $\xbc$ and $x_b$, which means $d(c, a) \ge d(a, b)$.
    Consider the instance $\instance' = (\voters, \candidates, 2, \succ')$, and metric $d'$, where there are at least $n/(1+\alpha)$ voters on $\xbc$, at least $n(\alpha - 1)/(1 + \alpha)$ voters on $\xab$, and at most $n/(1+\alpha)$ voters on $x_b$, see Figure \ref{fig:ub(b,b)_1}.
    
    Since
    $\dist(f(\succ), \instance; d ) > 1 + \alpha$,
    $\left|\voters((-\infty, \xbc))\right| > n/(1 +\alpha)$,
    $\left|\voters((-\infty, \xab))\right| \geq n\alpha/(1 + \alpha)$,
    and $x^* = x_b$, based on Lemma \ref{lem:move-voters-with-prop}, we have
    $$\dist(f(\succ), \instance'; d' ) > 1 + \alpha.$$

    \begin{figure}[t]
        \centering
        \begin{tikzpicture}[line cap=round,line join=round,>=triangle 45,x=1cm,y=1cm]
        
        \draw (3,1) node[anchor=center ] {$|\voters(x_b)|\le\frac{n}{1 + \alpha}$};
        \draw (0,1) node[anchor=center] {$|\voters(\xbc)| \ge \frac{n}{1 + \alpha}$};
        \draw (1.875,1.8) node[anchor=center] {$|\voters(\xab)|\ge\frac{n(\alpha - 1)}{1 + \alpha}$};
        
        \draw [line width=2pt] (-3,0)-- (3,0);
        
        \draw [fill=black] (-3,0) circle (2pt);
        \draw[color=black] (-3,-0.5) node {$c$};
        \draw [fill=black] (3,0) circle (2pt);
        \draw[color=black] (3,-0.5) node {$b$};
        \draw [fill=black] (0,0) circle (2pt);
        \draw[color=black] (0,-0.5) node {$\xbc$};
        \draw [fill=black] (0.75,0) circle (2pt);
        \draw[color=black] (0.75,-0.5) node {$a$};
        \draw [fill=black] (1.875,0) circle (2pt);
        \draw[color=black] (1.875,-0.5) node {$\xab$};
        
        \draw[->]  (0, 0.7) to (0, 0.1);
        \draw[->]  (1.875, 1.5) to (1.875, 0.1);
        \draw[->]  (3, 0.7) to (3, 0.1);
        
        \end{tikzpicture}\caption{
        The instance for \ref{case:1}, when the optimal candidates are both located on $x_b$ and $d(c, a) \ge d(a, b)$. The output committee is $\{ a, c\}$. In the new metric $d'$, the voters are distributed across three locations, as illustrated in the figure.}
        \label{fig:ub(b,b)_1}
    \end{figure}

    Now it remains to calculate distortion of committee $\{a, c\}$ in $\instance'$. For brevity, let $r=d(a, \xbc)$, $s = d(\xab, a)$. Note that $d(\xab, b ) = s$ and $d(\xbc, c) = 2s + r$. We can calculate $\dist(f(\succ), \instance'; d')$ as follows:
        
    \begin{align*}
        \dist(f(\succ), \instance'; d') &= \frac{\SC(a; d') + \SC(c; d')}{2\SC(b; d')}\\ 
        &= \dfrac{\frac{2s}{1+\alpha} + \frac{s(\alpha - 1)}{1+\alpha} + \frac{r}{1+\alpha} + (2s+r)+\frac{(r+s)\alpha}{1+\alpha} + \frac{s}{1+\alpha}}{2(\frac{s\alpha}{1+\alpha} + \frac{r+s}{1+\alpha})}\\
        &=\dfrac{4s + 2r}{2(s + \frac{r}{1+\alpha})}\\
        &\leq 1 + \alpha\\
        & = 1 + \sqrt{2}.
    \end{align*}
        
    This completes the proof for the subcase that $x_a$ resides between $\xbc$ and $x_b$.

    In the other subcase, the output of our voting rule is the committee $\{a,c\}$, but $x_a$ is located on the left side of $\xbc$, indicating that $d(c, a) < d( a, b)$. We again consider the instance $\instance' = (\voters, \candidates, 2, \succ')$ and metric $d'$, where there are at least $n/(1+\alpha)$ voters on $\xbc$, at least $n(\alpha - 1)/(1 + \alpha)$ voters on $\xab$, and at most $n/(1+\alpha)$ voters on $x_b$, see Figure \ref{fig:ub(b,b)_2}. 
    Since $\dist(f(\succ), \instance; d ) > 1 + \alpha$,
    also $\left|\voters((-\infty, \xbc))\right| > n/(1 +\alpha)$,
    $\left|\voters((-\infty, \xab))\right| \geq n\alpha/(1 + \alpha)$,
    and $x^* = x_b$, based on Lemma \ref{lem:move-voters-with-prop}, we have 
    $$\dist(f(\succ), \instance'; d' ) > 1 + \alpha.$$
        
    Now it remains to calculate distortion of committee $\{a, c\}$ in $\instance'$ and $d'$. For brevity, let $s = d(\xab, b)$, and $r = d(\xbc, a)$. One can conclude that $d(a, \xab) = s$, and $d(c, \xbc) = 2s - r$. Therefore,
    \begin{align*}
        \dist(f(\succ), \instance'; d') &= \frac{\SC(a; d') + \SC(c; d')}{2\SC(b; d')}\\ 
        & \leq \dfrac{\left( r + (s-r)(\frac{\alpha}{1 + \alpha}) + s(\frac{1}{1 + \alpha})\right) \times 2 + (2s - 2r)}{2(s(\frac{\alpha}{1 + \alpha})+ (s-r)\frac{1}{1 + \alpha})}\\
        & = \dfrac{(s-r)\frac{\alpha}{1+\alpha} + \frac{s}{1 + \alpha} + s}{s(\frac{\alpha}{1 + \alpha}) + (s-r)\frac{1}{1+\alpha}} \\
        &\leq \dfrac{(s-r) + s\alpha}{(s-r)\frac{1}{1+\alpha} + s(\frac{\alpha}{1 + \alpha})} \\
        & = 1 + \alpha\\
        &= 1 + \sqrt{2}.
\end{align*}
    This completes the proof for this case.

\begin{figure}[t]
\centering
\begin{tikzpicture}[line cap=round,line join=round,>=triangle 45,x=1cm,y=1cm]

\draw (3,1) node[anchor=center] {$|\voters(x_b)|\le\frac{n}{1 + \alpha}$};
\draw (0,1) node[anchor=center] {$|\voters(\xbc)|\ge\frac{n}{1 + \alpha}$};
\draw (1.5,1.8) node[anchor=center] {$|\voters(\xab)|\ge \frac{n(\alpha - 1)}{1 + \alpha}$};

\draw [line width=2pt] (3,0)-- (-3,0);

\draw [fill=black] (-3,0) circle (2pt);
\draw[color=black] (-3,-0.5) node {$c$};
\draw [fill=black] (3,0) circle (2pt);
\draw[color=black] (3,-0.5) node {$b$};
\draw [fill=black] (0,0) circle (2pt);
\draw[color=black] (0,-0.5) node {$\xbc$};
\draw [fill=black] (-0.5,0) circle (2pt);
\draw[color=black] (-0.5,-0.5) node {$a$};
\draw [fill=black] (1.5,0) circle (2pt);
\draw[color=black] (1.5,-0.5) node {$\xab$};

\draw[->]  (0, 0.7) to (0, 0.1);
\draw[->]  (1.5, 1.5) to (1.5, 0.1);
\draw[->]  (3, 0.7) to (3, 0.1);

\end{tikzpicture}\caption{
A modified instance for \ref{case:1}, when the optimal candidates are $x_b$ and $d(c, a) < d(a, b)$. The output committee is $\{ a, c\}$. In the new metric $d'$, voters are distributed across three locations, as illustrated in the figure.}
\label{fig:ub(b,b)_2}
\end{figure}

\end{proofcase}

\item \label{case:2} $(x_{\optcandidate_1}, x_{\optcandidate_2}) = (x_c, x_c)$.

\begin{proofcase}
    Our voting rule first selects $a$ and then chooses between $c$ and $b$. The case that $\{a, c\}$ are selected is similar to \ref{case:1} and results in distortion at most 2. For the case that $\{a, b\}$ is selected, we know that $|V_{c \succ b}| \le n/(1 + \alpha)$. Furthermore, by Corollary \ref{col:opt-position}, $\midd$ must lie between $c$ and $a$. Thus, by Lemma \ref{lem:relative-sc}, $\SC(a) \le \SC(b)$. Therefore, the distortion of instance $\instance$ in this case is
    \begin{align*}
        \dist(f(\succ), \instance; d) &= \dfrac{\SC(a) + \SC(b)}{2 \SC(c)} \\
        &\le \dfrac{2 \SC(b)}{2\SC(c)} \\
        &= \dfrac{\SC(b)}{\SC(c)}.
    \end{align*}
    Since $|V_{c \succ b}| \le n/(1 + \alpha)$, we have $|V_{b \succ c}| > n\alpha/(1 + \alpha)$. Thus, based on Lemma \ref{lem:ratio-alters}, 
    \begin{align*}
    \dist(f(\succ), \instance; d) &\le \frac{\SC(b)}{\SC(c)} \\
    &\leq \dfrac{\ \ \ 2n\ \ \ }{\frac{n\alpha}{1+\alpha}} - 1 \\
    &= 1 + \frac{2}{\alpha} \\
    &= 1 + \sqrt{2}.
    \end{align*}

\end{proofcase}

\item \label{case:3} $(x_{\optcandidate_1}, x_{\optcandidate_2}) \in \{(x_a, x_b), (x_b, x_a) \}$.

\begin{proofcase}
    The outcome would be optimal if our voting rule chooses $b$ over $c$. Thus, assume that the output of our voting rule is $\{a, c\}$. By the construction of Algorithm \ref{alg:2-winner-voting-rule}, this happens when $|V_{c \succ b}| > n/(1 + \alpha)$. Therefore, by Observation \ref{obs:set-of-voters-diff-view}, we have that $\left|\voters((-\infty, \xbc))\right| > n/(1 +\alpha)$. In this instance, the distortion would be 
    \[\dist(f(\succ), \instance; d) = \dfrac{\SC(a) + \SC(c)}{\SC(a) + \SC(b)} .\]
    Note that for this case, based on Algorithm \ref{alg:move-voters-finding-point}, the \focalpoint would be the point $x_b$.
    Consider the instance $\instance' = (\voters, \candidates, 2, \succ')$, and metric $d'$, where there are at least $n/(1+\alpha)$ voters on $\xbc$, and at most $n\alpha/(1+\alpha)$ voters on $x_b$, see Figure \ref{fig:ub(a,b)}. Since $\dist(f(\succ), \instance; d ) > 1 + \alpha$, $\left|\voters((-\infty, \xbc))\right| > n/(1 +\alpha)$, and $x^* = x_b$, based on Lemma \ref{lem:move-voters-with-prop}, we can conclude that 
    $$\dist(f(\succ), \instance'; d' ) > 1 + \alpha.$$

Now, it remains to calculate the distortion of committee $\{a, c\}$. For brevity, let $s = d(a, b)$, and $r = d(\xbc, a)$. Note that $d(c, \xbc) = r + s$.
Therefore,
\begin{align*}
\dist(f(\succ), \instance'; d') &= \dfrac{\SC(a; d') + \SC(c; d')}{\SC(a; d') + \SC(b; d')} \\
&\le \frac{s\alpha + r + (r + s)(2\alpha + 1)}{s\alpha + r + r + s} \\
&= 1 + \dfrac{2(r +s)\alpha}{s(\alpha + 1) + 2r} \\
& \le 1 + \dfrac{2(r +s) \alpha}{2s + 2r} \\
&= 1 + \alpha\\
& = 1 + \sqrt{2},
\end{align*}
which completes the proof.

\begin{figure}[t]
\centering
\begin{tikzpicture}[line cap=round,line join=round,>=triangle 45,x=1cm,y=1cm]

\draw (3,1) node[anchor=center] {$|\voters(x_b)|\le \frac{n\alpha}{1 + \alpha}$};
\draw (0,1) node[anchor=center] {$|\voters(\xbc)|\ge\frac{n}{1 + \alpha}$};

\draw [line width=2pt] (3,0)-- (-3,0);

\draw [fill=black] (-3,0) circle (2pt);
\draw[color=black] (-3,-0.5) node {$c$};
\draw [fill=black] (3,0) circle (2pt);
\draw[color=black] (3,-0.5) node {$b$};
\draw [fill=black] (0,0) circle (2pt);
\draw[color=black] (0,-0.5) node {$\xbc$};
\draw [fill=black] (0.5,0) circle (2pt);
\draw[color=black] (0.5,-0.5) node {$a$};

\draw[->]  (0, 0.7) to (0, 0.1);
\draw[->]  (3, 0.7) to (3, 0.1);

\end{tikzpicture}\caption{
A modified instance for \ref{case:3}. The optimal candidates are located on $(x_a, x_b)$. The output committee is $\{ a, c\}$. In the new metric $d'$, the voters are distributed across two locations, as illustrated in the figure.}
\label{fig:ub(a,b)}
\end{figure}

\end{proofcase}

\item \label{case:4} $(x_{\optcandidate_1}, x_{\optcandidate_2}) \in \{ (x_a, x_c), (x_c, x_a)\}$.

\begin{proofcase}
    The outcome would be optimal if our voting rule chooses $c$ over $b$. Thus, assume that the output of our voting rule is $\{a, b\}$. By the construction of the Algorithm \ref{alg:2-winner-voting-rule}, this happens either when 
    $|V_{c \succ b}| \le n/(1 + \alpha)$, or when $|V_{c \succ b}| > n/(1 + \alpha)$ and $|V_{b\succ a}| \ge n/(1 + \alpha)$. For now, let us consider the first subcase that $|V_{c \succ b}| \le n/(1 + \alpha)$. 
    
    In this instance, the distortion would be:
    \begin{align*}
        \dist(f(\succ), \instance; d) &= \dfrac{\SC(a) + \SC(b)}{\SC(a) + \SC(c)}\\
        &\leq \frac{\SC(b)}{\SC(c)}\\
        &\leq \dfrac{\ \ \ 2n\ \ \ }{\frac{n\alpha}{1+\alpha}} - 1 \\
        &= 1 + \dfrac{2}{\alpha} \\
        & = 1 + \sqrt{2}.
    \end{align*}
    Note that the second inequality follows from Lemma \ref{lem:ratio-alters}, since $|V_{b \succ c}| \geq n\alpha/(1 + \alpha)$.

    Now consider the second subcase, where $|V_{c \succ b}|>n/(1 + \alpha)$, and also $|V_{b\succ a}| \ge n/(1 + \alpha)$. By Observation \ref{obs:set-of-voters-diff-view}, it follows that $\left|\voters(\left[\xab, \infty\right))\right| \geq n/(1 + \alpha)$. Furthermore, since $b \succ_\midd c$, thus $\left|\voters(\left[\xbc, \infty\right))\right| \geq n/2$.
    In this case, the output of our voting rule is $\{a, b\}$, and based on Algorithm \ref{alg:move-voters-finding-point}, the \focalpoint would be the point $x_c$. 
    Consider the instance $\instance' = (\voters, \candidates, 2, \succ')$, and metric $d'$, where $|\voters(x_c)| \le n/2$, $|\voters(\xab)| \ge n/(1 + \alpha)$, and $|\voters(\xbc)| \ge n\left(1/2 - 1/(1 + \alpha)\right)$, see Figure \ref{fig:ub(a,c)_2}.
    Since $\dist(f(\succ), \instance; d ) > 1 + \alpha$, $\left|\voters(\left[\xab, \infty\right))\right| \geq n/(1 + \alpha)$, $\left|\voters(\left[\xbc, \infty\right))\right| \geq n/2$, and $x^* = x_c$, based on Lemma \ref{lem:move-voters-with-prop}, we have 
    $$\dist(f(\succ), \instance'; d' ) > 1 + \alpha.$$

    Now it remains to calculate distortion of committee $\{a, b\}$ in $\instance'$. For brevity, let $r = d(a, c)$, $t = d(a, \xbc)$, and $s = d(\xbc, \xab)$. Note that $d(\xab, b) = s + t$. Therefore,
    \begin{align*}
        \dist&(f(\succ), \instance'; d') = \dfrac{\SC(a; d') + \SC(b; d')}{\SC(a; d') + \SC(c; d')}\\
        &\le \dfrac{n(s+t) + n(\frac{1}{2} - \frac{1}{1 + \alpha})s + \frac{n}{2}(r + s + t) + \frac{nr}{2} + \frac{nt}{2} + \frac{ns}{1 + \alpha}}{\frac{nr}{2} + \frac{nt}{2} + \frac{ns}{1 + \alpha} + \frac{n(r+t)}{2} + \frac{ns}{1 + \alpha}}\\
        & \le \frac{2s + 2t +r}{r + t + \frac{2s}{1 + \alpha}} \\
        &\le 1 + \alpha\\
        & = 1 + \sqrt{2}.
    \end{align*}
    This completes the proof.

    \begin{figure}[t]
        \centering
        \begin{tikzpicture}[line cap=round,line join=round,>=triangle 45,x=1cm,y=1cm]
        
        \draw [line width=2pt] (3,0)-- (-3,0);
        
        \draw (-3,1) node[anchor=center] {$|\voters(x_c)|\le\frac{n}{2}$};
        \draw (1.5,1) node[anchor=center] {$|\voters(\xab)|\ge\frac{n}{1 + \alpha}$};
        \draw (0,1.7) node[anchor=center] {$|\voters(\xbc)|\ge n \left( \frac{1}{2} - \frac{1}{1 + \alpha} \right)$};

        \draw [fill=black] (-3,0) circle (2pt);
        \draw[color=black] (-3,-0.5) node {$c$};
        \draw [fill=black] (3,0) circle (2pt);
        \draw[color=black] (3,-0.5) node {$b$};
        \draw [fill=black] (-0.5,0) circle (2pt);
        \draw[color=black] (-0.5,-0.5) node {$a$};
        \draw [fill=black] (1.5,0) circle (2pt);
        \draw[color=black] (1.5,-0.5) node {$\xab$};
        \draw [fill=black] (0,0) circle (2pt);
        \draw[color=black] (0,-0.5) node {$\xbc$};

        \draw[->]  (1.5, 0.7) to (1.5, 0.1);
        \draw[->]  (-3, 0.7) to (-3, 0.1);
        \draw[->]  (0, 1.4) to (0, 0.1);
        
        \end{tikzpicture}\caption{
        A modified instance for \ref{case:4}, when the optimal candidates are located on $(x_a, x_c)$, $|V_{c \succ b}| > n/(1 + \alpha)$, and also $|V_{b \succ a} | \ge n/(1 + \alpha)$. The output committee is $\{ a, b\}$. In the new metric $d'$, the voters are distributed across three locations, as illustrated in the figure.}
        \label{fig:ub(a,c)_2}
    \end{figure}

\end{proofcase}
        
\end{enumerate}

\end{proof}

We now show that any $2$-winner voting rule in the line metric, under a utilitarian additive social cost function, must have a distortion of at least $1 + \sqrt{2}$. This result matches our upper bound, making the distortion bound tight for $2$-winner voting.

\begin{theorem}\label{thm:lb2.41}
    Any $2$-winner voting rule in the line metric and under additive social cost, has a distortion of at least $1+ \sqrt{2}$.
\end{theorem}

\begin{proof}
    We construct an instance $\instance = (\voters, \candidates, k, \succ)$ with two different metrics $d_1$ and $d_2$. Let $X$ be a subset of voters of size $n_1 = n/(1 + \sqrt{2})$ whose preference profile is $a' \succ b' \succ a \succ b$. Similarly, let $Y$ be a subset of voters of size $n_2 = n\sqrt{2}/(1 + \sqrt{2})$ with the preference profile $a \succ b \succ a' \succ b'$. In both metrics, assume that candidates $a$ and $b$ reside at point $x = 1$, and candidates $a'$ and $b'$ reside at point $x = -1$. 
    In $d_1$, locate the voters in $X$ at point $x = -1$, and the remaining $n_2$ voters in $Y$ at point $x = 0$, as shown in Figure \ref{fig:lb-2-i1}.
    In $d_2$, locate the voters in $X$ at point $x = 0$, and the remaining $n_2$ voters in $Y$ at point $x = 1$, as shown in Figure \ref{fig:lb-2-i2}.

\begin{figure}[t]
    \centering
    \begin{subfigure}[t]{0.5\textwidth}
        \centering
        \begin{tikzpicture}[line cap=round,line join=round,>=triangle 45,x=1cm,y=1cm]
        
        \draw [line width=2pt,color=black] (-3,0)-- (3,0);

        \draw (0,1) node[anchor=center] {$|\voters(0)|=\frac{n\sqrt{2}}{1 + \sqrt{2}}$};
        \draw (-3,1) node[anchor=center] {$|\voters(-1)|=\frac{n}{1 + \sqrt{2}}$};
    
        \draw [fill=black] (3,0) circle (1.5pt);
        \draw[color=black] (2.5,-0.5) node {$+1, a, b$};
        
        \draw [fill=black] (-3,0) circle (1.5pt);
        \draw[color=black] (-3,- 0.45) node {$-1, a', b'$};
        
        \draw [fill=black] (0,0) circle (1.5pt);
        \draw[color=black] (0,-0.5) node {$0$};

        \draw[->]  (0, 0.7) to (0, 0.1);
        \draw[->]  (-3, 0.7) to (-3, 0.1);
        
        \end{tikzpicture}
        \caption{}
        \label{fig:lb-2-i1}
    \end{subfigure}%
    ~ 
    \begin{subfigure}[t]{0.5\textwidth}
        \centering
        \begin{tikzpicture}[line cap=round,line join=round,>=triangle 45,x=1cm,y=1cm]
        
        \draw [line width=2pt,color=black] (-3,0)-- (3,0);

        \draw (3,1) node[anchor=center] {$|\voters(1)|=\frac{n\sqrt{2}}{1 + \sqrt{2}}$};
        \draw (0,1) node[anchor=center] {$|\voters(0)|=\frac{n}{1 + \sqrt{2}}$};
    
        \draw [fill=black] (3,0) circle (1.5pt);
        \draw[color=black] (3,-0.5) node {$+1, a, b$};
        
        \draw [fill=black] (-3,0) circle (1.5pt);
        \draw[color=black] (-2.5,- 0.45) node {$-1, a', b'$};
        
        \draw [fill=black] (0,0) circle (1.5pt);
        \draw[color=black] (0,-0.5) node {$0$};

        \draw[->]  (0, 0.7) to (0, 0.1);
        \draw[->]  (3, 0.7) to (3, 0.1);
        
        \end{tikzpicture}
        \caption{}
        \label{fig:lb-2-i2}
    \end{subfigure}
    \caption{Instance $\instance$ with two metrics $d_1$  (a) and $d_2$ (b), used for proving the lower bound in $2$-winner election. $a, b, a',$ and $b'$ are the candidates distributed on locations $-1$ and $+1$. 
    }
    \label{Fig:Main}
\end{figure}
  
    Note that in both instances, the preference profile $\succ$ is consistent with the corresponding metric. If we apply a voting rule $f$ to $\succ$, we can consider three cases for the output:

- If $f$ selects $a'$ and $b'$, distortion with respect to $d_2$ would be:
  \begin{align*}
      \dist(f(\succ), \instance; d_2) &\geq \frac{\SC(\{a', b'\}; d_2)}{\SC(\{a, b\}; d_2)}\\
      &= \frac{\SC(a'; d_2) + \SC(b'; d_2)}{\SC(a; d_2) + \SC(b; d_2)}\\
      &= \frac{2 n_1 + 4 n_2}{2 n_1}\\ 
      &= 1 + 2\sqrt{2}.
   \end{align*}

- If $f$ selects $a$ and $b$, distortion with respect to $d_1$ would be:
  \begin{align*}
  \dist(f(\succ), \instance; d_1) &\geq \frac{\SC(\{a, b\}; d_1)}{\SC(\{a', b'\}; d_1)}\\
  &= \frac{\SC(a; d_1) + \SC(b; d_1)}{\SC(a'; d_1) + \SC(b'; d_1)}\\
  &= \frac{4 n_1 + 2 n_2}{2 n_2}\\
  &= 1 + \sqrt{2}.
  \end{align*}

- If $f$ selects $a$ and $a'$, distortion with respect to $d_1$ would be:
  \begin{align*}
  \dist(f(\succ), \instance; d_1) &\geq \frac{\SC(\{a, a'\}; d_1)}{\SC(\{a', b'\}; d_1)}\\
  &= \frac{\SC(a; d_1) + \SC(a'; d_1)}{\SC(a'; d_1) + \SC(b'; d_1)}\\
  &= \frac{2 n_1 + 2 n_2}{2 n_1}\\
  &= 1 + \sqrt{2}.
  \end{align*}

Therefore, for every possible output, the distortion in at least one of the two instances would be greater than or equal to $1 + \sqrt{2}.$

\end{proof}

\section{General Multi-Winner Voting}
\label{sec:multi-winner}
In this section, we generalize our voting rule to $k$-winner voting rules on the line metric under a utilitarian additive social cost function. We establish both upper and lower bounds on the distortion for these rules.
First, we present in Theorem \ref{thm:ub-general} an upper bound on the distortion of $k$-winner voting rules, which is approximately $7/3$ but may be slightly higher when $k$ is not divisible by $3$.
Additionally, we derive two lower bounds. For $k \geq m/2$, the distortion of any $k$-winner voting rule is lower bounded by $1 + (m - k)/(3k - m)$. For $k \leq m/2$, the distortion is lower bounded by $2 + 1/k$ when $k$ is odd and $1 + \sqrt{1 + 2/k}$ when $k$ is even. These results provide an overview of distortion behavior for $k$-winner voting rules across various values of $k$.


\begin{theorem} \label{thm:ub-general}
For every integer $k > 2$, there exists a $k$-winner voting rule $f$ such that for any election instance $\instance = (\voters, \candidates, k, \succ)$ where voters and candidates lie on a line metric, and for any metric $d$ consistent with $\succ$, the distortion of $f$ under the utilitarian additive social cost satisfies:
\[
\dist(f) \le 
\begin{cases}
    7/3 & \text{for } k \equiv 0 \pmod{3},\\
    7/3 + \frac{4(\sqrt{2} - 4/3)}{k} & \text{for } k \equiv 1 \pmod{3},\\
    7/3 + \frac{2(\sqrt{2} - 4/3)}{k} & \text{for } k \equiv 2 \pmod{3}.
\end{cases}
\]
\end{theorem}

Before proving this theorem, we demonstrate how combining two multi-winner voting rules, each selecting different committee sizes, yields a $k$-winner voting rule with low distortion in Subsection \ref{Combining Voting Rules}. Subsequently, in Subsection \ref{Polar Comparison Rule for $k=3$}, we generalize the Polar Comparison Rule for $k=3$, to achieve a low distortion.

\subsection{Combining Voting Rules}\label{Combining Voting Rules}
In this subsection, we present a method for combining different voting rules to construct a $k$-winner selection rule. We show that the distortion of the combined rule can be expressed as a function of the distortions of the individual voting rules.

\begin{theorem}\label{thm:ub-linear-combination}
Given integers $k, k_1, k_2, \alpha, \beta$ where $k = \alpha k_1 + \beta k_2$, if $f_1$ and $f_2$ are respectively $k_1$-winner and $k_2$-winner voting rules, with distortions at most $\dd_1$ and $\dd_2$, there exists a $k$-winner voting rule with distortion at most
$(\alpha k_1 \dd_1 + \beta k_2\dd_2)/k.$
\end{theorem}

\begin{proof}
    For the simpler case where $\beta = 0$, by running $f_1$ for $k/k_1$ iterations, each time removing the candidates that were selected as winners, we can achieve a distortion of at most $\dd_1$ for the $k$-winner election, as shown by \cite{10.1145/3230654.3230658}. 
    
    For the general case, we use a similar approach: assuming $\dd_2 \le \dd_1$ without loss of generality, we fix an instance $\instance = (\voters, \candidates, k, \succ)$. Define voting rule $f$ as follows: First, run $f_1$ iteratively to select $\alpha k_1$ candidates $S_1 = \{ a_1, \dots a_{\alpha k_1}\}$. Subsequently, apply $f_2$ iteratively on the remaining candidates $\candidates \backslash S_1$, to select $\beta k_2$ candidates $S_2 = \{ b, \dots , b_{\beta k_2}\}$. We prove that this voting rule, $f$, has the desired distortion for $k$-winner elections.

    Recall that $\opt = \{ \optcandidate_1, \optcandidate, \dots ,\optcandidate_k\}$ is the optimal committee of size $k$, and $\SC(\optcandidate_i) \le \SC(\optcandidate_j)$ for all $i \le j$. Partition this set into two subsets $\opt_1 = \{ \optcandidate_i \mid 1 \le i \le \alpha k_1\}$ and $\opt_2 = \{ \optcandidate_j \mid \alpha k_1 < j \le k \}$.
    By the assumptions of the theorem, we know that 
    \begin{equation}\label{eq:s1<dopt1}
    \SC(S_1) \le \dd_1 \SC(\opt_1).
    \end{equation}
    Furthermore, when running $f_2$, $\alpha k_1$ candidates have been eliminated from $\candidates$. Thus, if $\opt_3$ is the optimal committee of size $\beta k_2$ in the instance given to $f_2$, $\SC(\opt_3) \le \SC(\opt_2)$. By the assumptions of the theorem, $\SC(S_2) \le \dd_2 \SC(\opt_3)$. Therefore, 
    \begin{equation}\label{eq:s2<dopt2}
    \SC(S_2) \le \dd_2 \SC(\opt_2).
    \end{equation}
    Using equations \eqref{eq:s1<dopt1} and \eqref{eq:s2<dopt2}, we can conclude that
    \begin{align*}
        \SC(S_1) + \SC(S_2) & \le  \dd_1 \SC(\opt_1) +  \dd_2 \SC(\opt_2)\\
        & = \dd_2 (\SC(\opt_1) + \SC(\opt_2)) + (\dd_1 - \dd_2 ) \SC(\opt_1)\\
        & \le \dd_2 \SC(\opt) + (\dd_1 - \dd_2) \dfrac{\alpha k_1 \SC(\opt)}{k} \\
        & = \dfrac{ \beta k_2\dd_2 + \alpha k_1\dd_1}{k} \SC(\opt),
    \end{align*}
 which completes the proof.
\end{proof}

\subsection{Polar Comparison Rule for $k=3$}\label{Polar Comparison Rule for $k=3$}
In this subsection, we establish an upper bound of $7/3$ for distortion in the case $k=3$. We show that this bound is tight, matching the lower bound derived in Theorem \ref{thm:lb-general}. We first define a voting rule by generalizing Polar Comparison rule for $k=3$, and then prove the $7/3$ upper bound on its distortion in Theorem \ref{thm:ub-3}. 

In Algorithm \ref{alg:3-winner-voting-rule}, we present a voting rule for $3$-winner elections, referred to as the \emph{Polar Comparison Rule} for the case $k=3$.
In this voting rule, the first two members of the committee are $a$ and $b$, the top two candidates in the \majorder. To select the third member, we compare the two candidates adjacent to $a$ and $b$, denoted by $b'$ and $c$.
Without loss of generality, assume that $b$ lies on the right side of $a$. Note that we can obtain the ordering of candidates that are not Pareto-dominated by any candidate other than $a$ using Lemma \ref{lem:candidates-order}. Therefore, we can find the closest candidate to the left side of $a$, namely $c$, and the closest candidate to the right of $b$, namely $b'$. As stated before, the goal is to either select $\{a, b, b' \}$ or $\{ a, b, c\}$ as the output committee. Therefore, similar to Polar Comparison rule for $2$-winner elections described in Algorithm \ref{alg:2-winner-voting-rule}, we compare $c$ and $b'$ to select the final candidate. Specifically, if $|\voters_{c\succ b'}| \ge 2n/5$, we choose $c$, otherwise, we choose $b'$.

\begin{nscenter}
\begin{algorithm}[t]
    \caption{Polar Comparison Rule for $3$-Winner Elections}
    \label{alg:3-winner-voting-rule}

    \newcommand\mycommfont[1]{\footnotesize\ttfamily\color{blue}#1}%
    \SetCommentSty{mycommfont}
    
    \SetKwInOut{Input}{Input}
    \SetKwInOut{Output}{Output}
    
    \Input{Instance $\instance = (\voters, \candidates, 3, \succ)$ }
        
    \Output{A Committee of three candidates $S = \{s_1, s_2, s_3\}$}

    \medskip
    $a, b \gets$ the top two candidates in the \majorder\;
    $b \gets \{e\in \candidates \mid \exists e'\in \candidates\setminus\{e, a\}: |\voters_{e'\succ e}| =n \}$  \textcolor{blue}{\tcp*{Candidates Pareto-dominated by $e'\ne a$}}
    $b' \gets \{e\in \candidates \mid \exists e'\in \candidates\setminus \{e, b\}: |\voters_{e'\succ e}| =n \}$ \hspace{-2pt}\textcolor{blue}{\tcp*{Candidates Pareto-dominated by $e'\ne b$}}
    $\candidates'_1 \gets \candidates\setminus ( b \cup \{a\})$\;
    $\candidates'_2 \gets \candidates\setminus  (b'\cup \{b\})$\;    
    Find $c \in \candidates_1'$ minimizing $d(c, a)$ such that $a$ lies between $c$ and $b$
    \textcolor{blue}{\tcp*{Using Lemma \ref{lem:candidates-order}}}
    Find $b' \in \candidates_2'$ minimizing $d(b', b)$ such that $b$ lies between $a$ and $b'$
    \textcolor{blue}{\tcp*{Using Lemma \ref{lem:candidates-order}}} 
    
    \If{$c$ does not exist}{
        \Return $\{ a, b, b'\}$\;
    }
    \If{$b'$ does not exist }{
        \Return $\{ a, b, c\}$ \;
    }
    \eIf{$|\voters_{c \succ b'}| \ge \frac{2n}{5}$}{
      \Return $\{a, b, c\}$\;
    }{\Return $\{ a, b, b'\}$\;
    }
    
\end{algorithm}
\end{nscenter}

\begin{theorem}\label{thm:ub-3}
If $f$ is the voting rule that selects a committee of size three by Algorithm \ref{alg:3-winner-voting-rule}, then $\dist(f) \le 7/3$.
\end{theorem}

\begin{proof}
To prove the theorem, we begin by assuming the contrary, i.e., there exists an instance $\instance = (\voters, \candidates, 3, \succ)$ and metric $d$, such that $\dist(f(\succ), \instance; d) > 7/3$. Let $a$ and $b$ be the first and second candidates in the \majorder of $\instance$, respectively. Without loss of generality, assume that $x_{b} \ge x_a$. Moreover, let $c$ be the closest candidate to the left side of $a$, and let $b'$ be the closest candidate to the right side of $b$, i.e., $x_c \le x_a \le x_{b} \le x_{b'}$.

The voting rule $f$, defined by Algorithm \ref{alg:3-winner-voting-rule}, first outputs candidates $a$ and $b$, and then selects one between $b'$ and $c$. Considering different cases for these two candidates and the optimal committee, we compute the distortion of our voting rule for each case separately and prove that it cannot exceed $7/3$.

\begin{enumerate}[label=\textbf{Case.\arabic*},ref=Case \arabic*, align=left]

\item \textbf{Candidate $c$ defeats $b'$}
\begin{proofcase}
    This implies that $|\voters_{c \succ b'}| \ge n/2$. Thus, $b'$ appears later than $c$ in the \majorder. Since $|\voters_{c\succ b} | \ge n/2$, the output of Algorithm \ref{alg:3-winner-voting-rule} would be $\{ a, b, c\}$. Therefore, 
    \[\dist(f(\succ), \instance) = \dfrac{\SC(a) + \SC(b) + \SC(c)}{\SC(\optcandidate_1) +\SC(\optcandidate_2) +\SC(\optcandidate_3)}.\]
    Note that $a$ might be in the optimal committee $\opt$ or not. If $a \not \in \opt$, by Corollary \ref{col:pm<3opt} $\SC(a) \le 3\SC(\optcandidate_3)$. Moreover, the median voter $\midd$, lies between $c$ and $b$. Therefore, if we hypothetically remove $a$ from the candidates, by Lemma \ref{lem:ub-3-le2}, $\SC(b) + \SC(c) \le 2(\SC(\optcandidate_1) +  \SC(\optcandidate_2))$. Thus,
    \begin{align*}
    \dist(f(\succ), \instance) &\le \dfrac{3\SC(\optcandidate_1) +2\SC(\optcandidate_2) +2\SC(\optcandidate_3)}{\SC(\optcandidate_1) +\SC(\optcandidate_2) +\SC(\optcandidate_3)}\\
    & \le \dfrac{5 \SC(\optcandidate_1) + 2\SC(\optcandidate_2)}{2\SC(\optcandidate_1) + \SC(\optcandidate_2)} \\
    &\le \dfrac{5 \SC(\optcandidate_1) + 2\SC(\optcandidate_2)}{(2 + 1/7) \SC(\optcandidate_1) + (1 - 1/7)\SC(\optcandidate_2)}\\
    &\le \dfrac{7}{3}.
    \end{align*}
    If $a\in \opt$, again by eliminating $a$ from the candidates and using Lemma \ref{lem:ub-3-le2}, we have 
    \begin{align*}
        \dist(f(\succ), \instance)  & =  \dfrac{\SC(a) + \SC(b) + \SC(c)}{\SC(\optcandidate_1) +\SC(\optcandidate_2) +\SC(\optcandidate_3)}\\
        & \le \dfrac{\SC(\optcandidate_1) + 2\left(\SC(\optcandidate_2) + \SC(\optcandidate_3)\right)}{\SC(\optcandidate_1) + \SC(\optcandidate_2) + \SC(\optcandidate_3)}\\
        & \le 2,
    \end{align*}
    which completes the proof of this case.
\end{proofcase}

\item \textbf{Candidate $b'$ defeats $c$}
\begin{proofcase}
    This means that $|\voters_{b' \succ c}| \ge n/2$. Thus, $b'$ appears earlier than $c$ in the \majorder. In this case, similar to the proof for $k=2$, we consider different cases for the optimal committee $\opt = \{ \optcandidate_1, \optcandidate_2, \optcandidate_3\}$. Note that by Lemma \ref{lem:relative-sc}, we only need to consider these cases for the optimal committee:
    \[ (x_{\optcandidate_1}, x_{\optcandidate_2}, x_{\optcandidate_3}) \in \{ 
    (x_{a}, x_{b}, x_{b'}),
    (x_{a}, x_{b}, x_{c}),
    (x_{a}, x_{c}, x_{c}),
    (x_{c}, x_{c}, x_{c}),
    (x_{b}, x_{b'}, x_{b'})
    \}.\]
    We consider each of these cases and prove the upper bound of $7/3$ separately.

\begin{enumerate}[label=\textbf{Case.2.\arabic*},ref=Case 2.\arabic*, align=left]

    \item  $(x_{\optcandidate_1}, x_{\optcandidate_2}, x_{\optcandidate_3}) = (x_{a}, x_{b}, x_{c})$.
    
    \begin{proofcase}
        Our voting rule, $f$, first selects $a, b$ and then chooses between $b'$ and $c$. If the output is $\{ a, b, c\}$, that would be the optimal committee with distortion 1. Otherwise, the output of $f$ is $\{ a, b, b'\}$. Therefore, the distortion can be expressed as
        \[\dist(f(\succ), \instance) = \dfrac{\SC(a) + \SC(b) + \SC(b')}{\SC(a) + \SC(b) + \SC(c)}.\]
        Note that by construction of the voting rule, $|\voters_{c\succ b'}|$ must not be greater than $2n/5$. Therefore, $|\voters_{b' \succ c}| \ge 3n/5$. Hence by Lemma \ref{lem:ratio-alters}, $\SC(b') \le 7/3\SC(c)$.
        Therefore,
        \[ \SC(a) + \SC(b) + \SC(b') \le  \frac{7}{3} \left(\SC(a) + \SC(b) +\SC(c)\right).\]
        This completes the proof for this case.
    \end{proofcase}
    
    \item  $(x_{\optcandidate_1}, x_{\optcandidate_2}, x_{\optcandidate_3}) = (x_{a}, x_{c}, x_{c})$.
    
    \begin{proofcase}
        Similar to the previous cases, $f$ chooses a candidate between $c$ and $b'$. If $c$ is selected, then,
        \[\dist(f(\succ), \instance) = \dfrac{\SC(a) + \SC(b) + \SC(c)}{\SC(a) + 2\SC(c)}.\]
        Since $b$ appears earlier than $c$ in the \majorder, by Corollary \ref{col:pm<3opt}, $\SC(b) \le 3 \SC(c)$. Therefore, it can be concluded that the above distortion is not greater than 2. 
        If $c$ is not selected,
        \[\dist(f(\succ), \instance) = \dfrac{\SC(a) + \SC(b) + \SC(b')}{\SC(a) + 2\SC(c)}.\]
        By the construction of the voting rule, this happens when $|\voters_{b' \succ c}| \ge 3n/5$. Similar to the previous case, using Lemma \ref{lem:ratio-alters}, we can conclude that $\SC(b') \le 7/3\SC(c)$. By Lemma \ref{lem:relative-sc}, $\SC(b) \le \SC(b')$. Therefore, $\SC(b) \le 7/3 \SC(c)$. These two facts demonstrate that the total distortion will not exceed $7/3$.
    \end{proofcase}
    
    \item  $(x_{\optcandidate_1}, x_{\optcandidate_2}, x_{\optcandidate_3}) = (x_{c}, x_{c}, x_{c})$.
    
    \begin{proofcase}
        The output of the voting rule is either $\{ a, b, b'\}$ or $\{ a, b, c\}$. Similar to the previous cases, the first output occurs when $|\voters_{b' \succ c}| \ge 3n/5$, which shows that $\SC(b') \le 7/3 \SC(c)$ by Lemma \ref{lem:ratio-alters}.
        Since $a \not \in \opt$ and $c\in \opt$, by Lemma \ref{lem:relative-sc}, one can conclude that $a$ and $c$ lie on different sides of $\midd$, the median voter. Therefore, again using Lemma \ref{lem:relative-sc}, $\SC(a) \le \SC(b) \le \SC(b')$. This implies that in this case,
        \begin{align*}
            \dist(f(\succ), \instance) &= \dfrac{\SC(a) + \SC(b) + \SC(b')}{3\SC(c)}\\
            &\le \dfrac{3\SC(b')}{3\SC(c)}\\
            &\le \dfrac{7}{3}. 
        \end{align*}
        Now consider the case that the output of $f$ is $\{ a, b, c\}$. Using Corollary \ref{col:pm<3opt}, one can conclude that $\SC(a) \le 3 \SC(c)$ and also $\SC(b) \le 3\SC(c)$. Therefore, 
        \begin{align*}
            \dist(f(\succ), \instance) &= \dfrac{\SC(a) + \SC(b) + \SC(c)}{3\SC(c)}\\
            &\le \dfrac{3\SC(c) + 3\SC(c) + \SC(c)}{3\SC(c)}\\
            &= \dfrac{7}{3}. 
        \end{align*}
        Thus, the proof of this case is completed.
    \end{proofcase}
    
\item $(x_{\optcandidate_1}, x_{\optcandidate_2}, x_{\optcandidate_3}) = (x_a, x_{b}, x_{b'})$. \label{case:ub_3_2.4}
    
\begin{proofcase}
    Our voting rule, $f$, first selects $a$ and $b$, and then chooses between $c$ and $b'$. In the case where $f$ selects $a, b,$ and $b'$, the output is the optimal committee with distortion $1$. Now consider the case that $f$ selects $c$. By the construction of Algorithm \ref{alg:3-winner-voting-rule}, this happens when $|\voters_{c \succ b'}| \geq 2n/5$. Therefore, by Observation \ref{obs:set-of-voters-diff-view}, it follows that $|\voters(\left(-\infty,\xcbtwo \right])| > 2n/5$.
    The distortion of $f$ for instance $\instance$ in this case is
    \begin{equation}\label{eq:ub_3_2.4}
    \dist(f(\succ), \instance) = \frac{\SC(a) + \SC(b) + \SC(c)}{\SC(a) + \SC(b) + \SC(b')}.    
    \end{equation}
    Note that for this case, based on Algorithm \ref{alg:move-voters-finding-point}, the \focalpoint would be $x_{b'}$.

    Suppose that we move $a$ and $b$ to $x_{\midd}$. Then, by Lemma \ref{lem:relative-sc}, their social cost would be reduced, which results in an increase in the distortion. Thereafter, consider the instance $\instance' = (\voters, \candidates, 2, \succ')$, and metric $d'$, where $|\voters(\xcbtwo)| = 2n/5$, $|\voters(\midd)| = n/10$, and $|\voters(x_{b'})| = n/2$, see Figure \ref{fig:ub-3-case2.4}. 
    Note that $x^* = x_{b'}$, and $\dist(f(\succ), \instance; d ) > 7/3$. Moreover,
    \begin{align*}
           &|\voters(\left(-\infty, \xcbtwo\right])| > 2n/5,\\
        & |\voters(\left(-\infty, \midd\right])| \geq n/2.
    \end{align*}
    Based on Lemma \ref{lem:move-voters-with-prop}, it follows that 
    $\dist(f(\succ), \instance'; d' ) > 7/3$.

    Now it remains to calculate the distortion of the committee $\{a, b, c\}$ in $\instance'$. For brevity, let $r=d(\xcbtwo, x_\midd)$, and $s = d(x_\midd, x_{b'})$. Note that $d(c, \xcbtwo ) = r+s$. We can calculate $\dist(f(\succ), \instance'; d')$ as follows:
        
    \begin{align*}
        \dist&(f(\succ), \instance'; d') = \frac{\SC(a) + \SC(b) + \SC(c)}{\SC(a) + \SC(b) + \SC(b')}\\ 
        &= 1 + \dfrac{r + s + \frac{r}{10} + \frac{r + s}{2} - (\frac{s}{2} + \frac{2r}{5})}{(\frac{s}{2} + \frac{2r}{5}) \times 2 + \frac{s}{2} + \frac{2r}{5}}\\
        &= 1 + \dfrac{\frac{6r}{5} + s}{\frac{6r}{5} + 1.5s}\\
        &\leq 2.
    \end{align*}
        
    This completes the proof for this case.

    \begin{figure}[t]
        \centering
        \begin{tikzpicture}[line cap=round,line join=round,>=triangle 45,x=1cm,y=1cm]
        
        \draw [line width=2pt,color=black] (-3,0)-- (3,0);
    
        \draw (1.4,1) node[anchor=center] {$|\voters(\midd)|= \frac{n}{10}$};
        \draw (3.8,1) node[anchor=center] {$|\voters(x_{b'})| = \frac{n}{2}$};
        \draw (-1,1) node[anchor=center] {$|\voters(\xcbtwo)| = \frac{2n}{5}$};
    
        \draw [fill=black] (3,0) circle (1.5pt);
        \draw[color=black] (3,-0.5) node {$b'$};
        
        \draw [fill=black] (-3,0) circle (1.5pt);
        \draw[color=black] (-3,- 0.45) node {$c$};
        
        \draw [fill=black] (0,0) circle (1.5pt);
        \draw[color=black] (-0.1,-0.5) node {$\xcbtwo$};
    
         \draw [fill=black] (1,0) circle (1.5pt);
        \draw[color=black] (1.1,-0.5) node {$\midd, a, b$};
        
        \draw[->]  (0, 0.7) to (0, 0.1);
        \draw[->]  (3, 0.7) to (3, 0.1);
        \draw[->]  (1, 0.7) to (1, 0.1);
        
        \end{tikzpicture}
        \caption{
        A modified instance for \ref{case:ub_3_2.4}. The optimal committee is $\{ a, b, b'\}$, and the output committee selected by $f$ is $\{a, b, c\}$. Candidates $a$ and $b$ are moved to $\midd$, the median voter. Moreover, the voters are distributed across three locations, $x_{\midd}, x_{b'}$, and $\xcbtwo$
        as illustrated in the figure.}
        \label{fig:ub-3-case2.4}
    \end{figure}
\end{proofcase}
    
\item $(x_{\optcandidate_1}, x_{\optcandidate_2}, x_{\optcandidate_3}) = (x_{b}, x_{b'}, x_{b'})$. \label{case:ub_3_2.5}

\begin{proofcase}        
    As stated before, $f$ first selects $a$ and $b$, and then chooses between $c$ and $b'$. If the voting rule outputs $\{a, b, b'\}$,
    $\SC(a) \leq 3\SC(b)$ and $\SC(b') \geq \SC(b)$. As a result, the distortion would be:
    \begin{align*}
    \dist(f(\succ), \instance) &= \frac{\SC(a) + \SC(b) + \SC(b')}{\SC(b) + 2\SC(b')} \\
    &\leq \frac{4\SC(b) + \SC(b')}{2\SC(b) + \SC(b')} \\
    &\leq \frac{5}{3},  
    \end{align*}
    which is less than $7/3$.
    Conversely, if $f$ selects $c$, by the construction of Algorithm \ref{alg:3-winner-voting-rule}, this happens when $|\voters_{c \succ b'}| \geq 2n/5$. Therefore, by Observation \ref{obs:set-of-voters-diff-view}, it follows that $|\voters(\left(-\infty, \xcbtwo\right])| \geq 2n/5$. Furthermore, since $b$ is the optimal candidate, by Corollary \ref{col:opt-position}, the median voter $\midd$ must lie between $a$ and $b$. 
    The distortion of $f$ in this case is
    \begin{equation}\label{eq:ub_3_2.5}
    \dist(f(\succ), \instance) = \frac{\SC(a) + \SC(b) + \SC(c)}{\SC(b) + 2\SC(b')}.
    \end{equation}
    Note that for this case, based on Algorithm \ref{alg:move-voters-finding-point}, the \focalpoint would be the point $x_{b'}$.

    Here we consider two cases. The first one is when $\xcbtwo$ is between $x_a$ and $x_\midd$.

    In the first case, note that since $a$ appears earlier than $b$ in the \majorder, $d(a, \midd) \le d(b, \midd)$. Furthermore, suppose that we move $b$ closer to $\midd$ until $d(b, \midd) = d(a, \midd)$. By Lemma \ref{lem:relative-sc}, this move does not decrease $\SC(b)$.
    Now, consider the instance $\instance' = (\voters, \candidates, 2, \succ')$, and a metric $d'$ where $|\voters(\xcbtwo)| = 2n/5$, $|\voters(\midd)| = n/10$, and $|\voters(x_{b'})| = n/2$, see Figure \ref{fig:ub-3-case2.5.1}. 
    Note that $\dist(f(\succ), \instance; d ) > 7/3$ and $x^* = x_{b'}$.
    Moreover, 
    \begin{align*}
         &|\voters(\left(-\infty, \xcbtwo\right])| > 2n/5,\\
         &|\voters(\left(-\infty, \midd\right])| \geq n/2.
    \end{align*}
    Therefore, by Lemma \ref{lem:move-voters-with-prop}, it follows that 
    $\dist(f(\succ), \instance'; d' ) > 7/3$.

    Now it remains to calculate the distortion of committee $\{a, b, c\}$ in $\instance'$. For brevity, let $t = d(c, a)$, $r=d(a, \xcbtwo)$, and $s = d(\xcbtwo, x_\midd)$. Note that $d(x_\midd, x_{b}) = r+s$, and $d(b, b') = t-2s$. We can calculate $\dist(f(\succ), \instance'; d')$ as follows:

    \begin{align*}
    \dist&(f(\succ), \instance'; d') = \frac{\SC(a) + \SC(b) + \SC(c)}{\SC(a) + \SC(b) + \SC(b')}\\ 
    &= \dfrac{
    r + \frac{s}{10} + \frac{t + r}{2} + 
    \frac{s+r}{2} + \frac{2s}{5} + \frac{t - 2s}{2} +
    r + t + \frac{s}{10} + \frac{t + r}{2}}{(\frac{s+r}{2} + \frac{2s}{5} + \frac{t - 2s}{2}) \times 3} \\
    &= \frac{1}{3} + \dfrac{2t + 3r + \frac{s}{5}}{\frac{3}{2} \times (t + r) - \frac{3s}{10}} \\
    &= \frac{7}{3} + \frac{-t + \frac{4s}{5}}{\frac{3}{2} \times (t + r) - \frac{3s}{10}} \\
    &\leq \frac{7}{3}.
    \end{align*} 
    The last inequality holds because given $d(b, b') \geq 0$, it follows that $t \geq 2s$.
    This completes the proof for the subcase that $\xcbtwo$ resides between $a$ and $x_\midd$.

    In the second subcase, the output of our voting rule is still the committee $\{a, b, c\}$, but $\xcbtwo$ is located on the left side of $a$. Similar to the first case, in this case, $d(a, \midd) \le d(b, \midd)$. Furthermore, we move $b$ closer to $\midd$ until $d(b, \midd) = d(a, \midd)$. By Lemma \ref{lem:relative-sc}, this will decrease $\SC(b)$.

    Now, consider the instance $\instance' = (\voters, \candidates, 2, \succ')$, and a metric $d'$ where $|\voters(\xcbtwo)| = 2n/5$, $|\voters(\midd)| = n/10$, and $|\voters(x_{b'})| = n/2$, see Figure \ref{fig:ub-3-case2.5.2}. 
    Note that $\dist(f(\succ), \instance; d) > 7/3$ and $x^* = x_{b'}$. Moreover, 
    \begin{align*}
        &|\voters(\left(-\infty, \xcbtwo\right])| > 2n/5,\\
        &|\voters(\left(-\infty, \midd\right])| \geq n/2.
    \end{align*}
    Therefore, based on Lemma \ref{lem:move-voters-with-prop}, it follows that $\dist(f(\succ), \instance'; d') > 7/3$.

    Now it remains to calculate the distortion of committee $\{a, b, c\}$ in $\instance'$. For brevity, let $t = d(\xcbtwo, a)$, $s=d(a, \midd)$, and $r = d(b, b')$. Note that $d(\midd, b) = s$, and $d(c, \xcbtwo) = t+2s+r$. We can calculate $\dist(f(\succ), \instance'; d')$ as follows:

    \begin{align*}
    \dist&(f(\succ), \instance'; d') = \frac{\SC(a) + \SC(b) + \SC(c)}{\SC(b) + 2\SC(b')}\\ 
    &= \dfrac{\frac{2t}{5} + \frac{s}{10} + \frac{2s + r}{2} + 
    (t + 2s + r) + \frac{t + s}{10} + \frac{t + 2s + r}{2}
    }{
    3\times(\frac{s + r}{2} + \frac{2(t+s)}{5})
    }\\
    &= \frac{1}{3} + \dfrac{2t + \frac{21s}{5} + 2r}{\frac{6t}{5} + \frac{27s}{10} + 1.5r}\\
    &\leq \frac{1}{3} + \frac{5}{3} \leq 2 < \frac{7}{3}.
    \end{align*}
    
    This completes the proof for this case.

\begin{figure}[t]
    \centering
    \begin{subfigure}[t]{0.5\textwidth}
        \centering
        \begin{tikzpicture}[line cap=round,line join=round,>=triangle 45,x=1cm,y=1cm]
        
        \draw [line width=2pt,color=black] (-3,0)-- (3,0);
    
        \draw (1.4,1) node[anchor=center] {$|\voters(\midd)|= \frac{n}{10}$};
        \draw (3.65,1) node[anchor=center] {$|\voters(x_{b'})| = \frac{n}{2}$};
        \draw (-1,1) node[anchor=center] {$|\voters(\xcbtwo)| = \frac{2n}{5}$};
    
        \draw [fill=black] (3,0) circle (1.5pt);
        \draw[color=black] (3,-0.5) node {$b'$};
        
        \draw [fill=black] (-3,0) circle (1.5pt);
        \draw[color=black] (-3,- 0.45) node {$c$};
        
        \draw [fill=black] (0,0) circle (1.5pt);
        \draw[color=black] (0,-0.5) node {$\xcbtwo$};
    
         \draw [fill=black] (0.8,0) circle (1.5pt);
        \draw[color=black] (0.8,-0.5) node {$\midd$};
    
        \draw [fill=black] (-0.6,0) circle (1.5pt);
        \draw[color=black] (-0.6,-0.5) node {$a$};
    
        \draw [fill=black] (2.4,0) circle (1.5pt);
        \draw[color=black] (2.4,-0.5) node {$b$};
        
        \draw[->]  (0, 0.7) to (0, 0.1);
        \draw[->]  (3, 0.7) to (3, 0.1);
        \draw[->]  (0.8, 0.7) to (0.8, 0.1);
        
        \end{tikzpicture}
        \caption{The point $\xcbtwo$ lies between $a$ and $\midd$.}
        \label{fig:ub-3-case2.5.1}
    \end{subfigure}%
    ~ 
    \begin{subfigure}[t]{0.5\textwidth}
        \centering
          \begin{tikzpicture}[line cap=round,line join=round,>=triangle 45,x=1cm,y=1cm]
        
        \draw [line width=2pt,color=black] (-3,0)-- (3,0);
    
        \draw (1.25,1) node[anchor=center] {$|\voters(\midd)|= \frac{n}{10}$};
        \draw (3.6,1) node[anchor=center] {$|\voters(x_{b'})| = \frac{n}{2}$};
        \draw (-1.1,1) node[anchor=center] {$|\voters(\xcbtwo)| = \frac{2n}{5}$};
    
        \draw [fill=black] (3,0) circle (1.5pt);
        \draw[color=black] (3,-0.5) node {$b'$};
        
        \draw [fill=black] (-3,0) circle (1.5pt);
        \draw[color=black] (-3,- 0.45) node {$c$};
        
        \draw [fill=black] (0,0) circle (1.5pt);
        \draw[color=black] (0,-0.5) node {$\xcbtwo$};
    
         \draw [fill=black] (1,0) circle (1.5pt);
        \draw[color=black] (1,-0.5) node {$\midd$};
    
        \draw [fill=black] (0.6,0) circle (1.5pt);
        \draw[color=black] (0.6,-0.5) node {$a$};
    
        \draw [fill=black] (1.6,0) circle (1.5pt);
        \draw[color=black] (1.6,-0.5) node {$b$};
        
        \draw[->]  (0, 0.7) to (0, 0.1);
        \draw[->]  (3, 0.7) to (3, 0.1);
        \draw[->]  (1, 0.7) to (1, 0.1);
        
        \end{tikzpicture}
        \caption{Candidate $a$ lies between $\xcbtwo$ and $\midd$.}
        \label{fig:ub-3-case2.5.2}
    \end{subfigure}
    \caption{Modified instances for the first (a) and second (b) subcases of \ref{case:ub_3_2.5}. The optimal candidates are located at $(x_{b}, x_{b'}, x_{b'})$, and the output committee selected by $f$ is $\{a, b, c\}$. The voters are distributed across three locations: $x_{\midd}$, $x_{b'}$, and $\xcbtwo$.}
\end{figure}

\end{proofcase}

\end{enumerate}
\end{proofcase}

\end{enumerate}
\end{proof}

Now, the proof of Theorem \ref{thm:ub-general} can be obtained as a corollary of theorems \ref{thm:line-voting-rule}, \ref{thm:ub-linear-combination}, and \ref{thm:ub-3}.

\begin{proof}[Proof of Theorem \ref{thm:ub-general}]
    Let $k_1 = 2$ and $k_2 = 3$. Moreover, let $f_1$ and $f_2$ be the Polar Comparison rules for $k=2$ and $k=3$, respectively. By Theorem \ref{thm:line-voting-rule}, $\dist(f_1) \le 1 + \sqrt{2}$, and by Theorem \ref{thm:ub-3}, $\dist(f_2) \le 7/3$.
    
    First, note that when $k=4$, $k = 2k_1$. Therefore, by Theorem \ref{thm:ub-linear-combination}, there exists a voting rule constructed by applying $f_1$ twice, with distortion at most $1 + \sqrt{2}$. 
    
    For general $k > 4$, express \( k \) as \( k = 3q + p \) with \( 0 \leq p < 3 \) and \( k \equiv p \pmod{3} \).
    If $p = 0$, $k = qk_2$. Hence, by Theorem \ref{thm:ub-linear-combination}, there exists a $k$-winner voting rule with distortion at most $7/3$. If $p = 1$, then $k = (q - 1)k_2 + 2k_1$. Applying Theorem \ref{thm:ub-linear-combination}, results in a $k$-winner voting rule with distortion at most 
    \[\dfrac{3(q-1)\times 7/3 + 4(1 + \sqrt{2})}{3q + 1}  = \frac{7}{3} + \frac{4(\sqrt{2} - 4/3)}{k}.\]
    If $p=2$, then $k = qk_2 + k_1$. Therefore by Theorem \ref{thm:ub-linear-combination}, there exists a voting rule with distortion at most 
    \[\dfrac{3q\times 7/3 + 2(1 + \sqrt{2})}{3q + 2}  = \frac{7}{3} + \frac{2(\sqrt{2} - 4/3)}{k}.\]
    This completes the proof. 

    \end{proof}

\subsection{Lower Bounds}\label{Lower Bounds}
After establishing the upper bounds, we now turn our attention to the lower bounds.
We provide two key results for lower bounds: one for $k \leq m/2$ and another for $k \geq m/2$.

\begin{theorem}\label{thm:lb-general}
    For any $k \leq m/2$, the distortion of any $k$-winner voting rule $f$ in the line metric, under additive social cost, can be lower bounded as:
    \[
    \dist(f) \geq
    \begin{cases}
        2 + \frac{1}{k} & \text{if } k \text{ is odd}, \\
        1 + \sqrt{1 + \frac{2}{k}} & \text{if } k \text{ is even}.
    \end{cases}
    \]
\end{theorem}

\begin{proof}
    Similar to the proof for $k = 2$, we construct an instance $\instance = (\voters, \candidates, k, \succ)$ 
    with two different metrics $d_1$ and $d_2$. 
    Assume there are two sets of candidates, $S_a$ and $S_b$, such that $|S_a|, |S_b| \geq k$. In both metrics, locate the candidates in $S_a$ at $x_1 = -1$ and the candidates in $S_b$ at $x_2 = +1$. For a fixed $k$, define the parameter $\alpha$ as follows:
    \[
    \alpha =
    \begin{cases}
        1 & \text{if } k \text{ is odd}, \\
        \sqrt{\frac{k}{k + 2}} & \text{if } k \text{ is even}.
    \end{cases}
    \]
To construct $\succ$, consider a preference profile with two voter groups. The \textit{small group}, consisting of $n_1 = n/(1 + \alpha) + \varepsilon$ voters, prefers all candidates in $S_a$ over those in $S_b$. The \textit{large group}, with $n_2 = n\alpha/(1 + \alpha) - \varepsilon$ voters, prefers all candidates in $S_b$ over those in $S_a$. The value $\varepsilon\in [0, 1]$ is chosen to ensure $n_1$ and $n_2$ are integers.

    For the first metric $d_1$, assume that the small group of voters lie at $x_1 = -1$ and the large group lie at $x_3 = 0$, as shown in Figure \ref{fig:lb-k-i1}. For the second metric $d_2$, place the small group at $x_3 = 0$ and the large group at $x_2 = 1$, as shown in Figure \ref{fig:lb-k-i2}.
    Note that in both metrics, the locations of voters are consistent with their preferences $\succ$. Furthermore, the optimal committee in $\instance$ with metric $d_1$ contains $k$ candidates from $S_a$, while the optimal committee in $\instance$ with metric $d_2$ consists of $k$ candidates from $S_b$.

\begin{figure}[t]
    \centering
    \begin{subfigure}[t]{0.5\textwidth}
        \centering
        \begin{tikzpicture}[line cap=round,line join=round,>=triangle 45,x=1cm,y=1cm]
        
        \draw [line width=2pt,color=black] (-3,0)-- (3,0);

        \draw (0.2,1) node[anchor=center] {$|\voters(0)|=\lfloor\frac{n\alpha}{1 + \alpha}\rfloor$};
        \draw (-2.8,1) node[anchor=center] {$|\voters(-1)|=\lceil\frac{n}{1 + \alpha}\rceil$};
    
        \draw [fill=black] (3,0) circle (1.5pt);
        \draw[color=black] (3,-0.5) node {$1, S_b$};
        
        \draw [fill=black] (-3,0) circle (1.5pt);
        \draw[color=black] (-3,- 0.5) node {$-1 , S_a$};
        
        \draw [fill=black] (0,0) circle (1.5pt);
        \draw[color=black] (0,-0.5) node {$0$};

        \draw[->]  (0, 0.7) to (0, 0.1);
        \draw[->]  (-3, 0.7) to (-3, 0.1);
        
        \end{tikzpicture}
        \caption{}
        \label{fig:lb-k-i1}
    \end{subfigure}%
    ~ 
    \begin{subfigure}[t]{0.5\textwidth}
        \centering
             \begin{tikzpicture}[line cap=round,line join=round,>=triangle 45,x=1cm,y=1cm]
        
        \draw [line width=2pt,color=black] (-3,0)-- (3,0);

        \draw (-0.2,1) node[anchor=center] {$|\voters(0)|=\lceil\frac{n}{1 + \alpha}\rceil$};
        \draw (2.8,1) node[anchor=center] {$|\voters(1)|=\lfloor\frac{n\alpha}{1 + \alpha}\rfloor$};
    
        \draw [fill=black] (3,0) circle (1.5pt);
        \draw[color=black] (3,-0.5) node {$1, S_b$};
        
        \draw [fill=black] (-3,0) circle (1.5pt);
        \draw[color=black] (-3,- 0.5) node {$-1 , S_a$};
        
        \draw [fill=black] (0,0) circle (1.5pt);
        \draw[color=black] (0,-0.5) node {$0$};

        \draw[->]  (0, 0.7) to (0, 0.1);
        \draw[->]  (3, 0.7) to (3, 0.1);
        
        \end{tikzpicture}
        \caption{}
        \label{fig:lb-k-i2}
    \end{subfigure}
    \caption{
    Instance $\instance$ with two metrics, $d_1$ (a) and $d_2$ (b), used to prove the lower bound on the distortion of $k$-winner voting rules for $k \le m/2$, as discussed in Theorem \ref{thm:lb-general}. $S_a$ and $S_b$ are two subsets of candidates located at positions $-1$ and $+1$. As shown in the figures, voters are divided into two groups of sizes $\lceil n/(1+\alpha)\rceil $ and $\lfloor n\alpha/(1+\alpha)\rfloor$.
    }
\end{figure}

    Given this instance and the two metrics, we prove a lower bound on the distortion of any voting rule $f$. Let $S = f(\succ)$ be the output committee selected by $f$. Assume that $S$ contains $r$ candidates from $S_a$ and $k - r$ candidates from $S_b$. One can observe that
    \begin{align*}
        \dist(S, \instance; d_1) &= 1 + 2\times \frac{k - r}{k}\times \frac{n_1}{n_2} , \\
        \dist(S, \instance; d_2) &= 1 + 2\times \frac{r}{k}\times \frac{n_2}{n_1}.
    \end{align*}

    First, consider the case that $r \geq (k + 1)/2$. If $k$ is odd, $\alpha=1$. Therefore, considering even values of $n$, we have $n_1=n_2$. Thus,
    \begin{align*}
    \dist(f) &\geq \dist(S, \instance; d_2)\\
    &= 1 + \frac{2r}{k}\\
    &\geq 1 + \frac{2(k + 1)}{2k} \\
    &= 2 + \frac{1}{k}.
    \end{align*}

    If $k$ is even, it implies that $r \geq k/2 + 1$. Therefore,
    \begin{align*}
        \dist(f) &\geq \dist(S, \instance; d_2) \\
        &= 1 + \frac{2r}{k} \times \frac{n_2}{n_1}\\
        & = 1 + \frac{2r}{k} \times \left(\alpha - \frac{(1+\alpha)^2\varepsilon}{n + (1+\alpha) \varepsilon}\right)\\
        & \ge 1 + \frac{2r}{k} \times (\alpha - \frac{4}{n})\\
        & \ge 1 + \frac{k+2}{k} \times \sqrt{\frac{k}{k + 2}} - \frac{8r}{nk}\\
        &\geq 1 + \sqrt{1 + \frac{2}{k}} - \frac{8}{n}.
    \end{align*}
    Note that as $n \to\infty$, $8/n \to 0$. Thus, given an even value for $k$, no voting rule can achieve a distortion better than $1 + \sqrt{1 + 2/k}$, which is the desired lower bound. 

    In the remaining case, we have $r < (k + 1)/2$. If $k$ is odd, then $r \leq (k - 1)/2$, and therefore,
    \begin{align*}
            \dist(f) &\geq \dist(S, \instance; d_1)\\
            &= 1 + \frac{2(k - r)}{k} \\
            &\geq 1 + \frac{k + 1}{k}\\
            &= 2 + \frac{1}{k}.
    \end{align*}

    If $k$ is even, it implies that $r \leq k/2$. Therefore,
    \begin{align*}
    \dist(f) &\geq \dist(S, \instance; d_1)\\
    &= 1 + 2 \times \frac{k - r}{k} \times \frac{n_1}{n_2}\\
    &\ge 1 + 2\times \frac{k-r}{k}  \times \sqrt{\frac{k + 2}{k}}\\
    &\geq 1 + \sqrt{1 + \frac{2}{k}},
    \end{align*}
    which completes the proof.
\end{proof}

\begin{theorem}\label{thm:lb->m/2}
    For any $k \ge m/2$, the distortion of any $k$-winner voting rule $f$ in the line metric, under additive social cost, can be lower bounded as
    $$\dist(f) \ge 1 + (m - k )/(3k - m).$$
\end{theorem}

\begin{proof}
    Similar to the proof of Theorem \ref{thm:lb-general}, we provide an instance $\instance= (\voters, \candidates, k, \succ)$ with two different metrics $d_1$ and $d_2$, having different locations for the voters. 
    Assume there are two sets of candidates, $S_a$ and $S_b$, such that $|S_a|= |S_b| =m/2$. In both metrics locate the candidates in $S_a$ at $x_1 = -1$ and the candidates in $S_b$ at $x_2 = +1$. 
    
    To construct $\succ$, consider a preference profile in which $n_1 = n/2$ voters prefer each candidate in $S_a$ to the candidates in $S_b$, and the remaining $n_2 = n/2$ voters prefer each candidate in $S_b$ to the candidates in $S_a$.
    For the first metric $d_1$, assume that the first group of voters are located at $x_1 = -1$ and the rest are located at $x_3 = 0$, as shown in Figure \ref{fig:lb-k>m/2-i1}. For the second metric $d_2$, place one group at $x_3 = 0$ and the other group at $x_2 = 1$, as shown in Figure \ref{fig:lb-k>m/2-i2}.
    Note that both metrics are consistent with the preference profiles $\succ$. Furthermore, the optimal committee in $\instance$ with metric $d_1$ contains $m/2$ candidates from $S_a$ and $k - m/2$ candidates from $S_b$, while the optimal committee in $\instance$ with metric $d_2$ consists of $m/2$ candidates from $S_b$ and $k - m/2$ candidates in $S_a$.

\begin{figure}[t]
    \centering
    \begin{subfigure}[t]{0.5\textwidth}
        \centering
        \begin{tikzpicture}[line cap=round,line join=round,>=triangle 45,x=1cm,y=1cm]
        
        \draw [line width=2pt,color=black] (-3,0)-- (3,0);

        \draw (0,1) node[anchor=center] {$|\voters(0)|=\frac{n}{2}$};
        \draw (-3,1) node[anchor=center] {$|\voters(-1)|=\frac{n}{2}$};
    
        \draw [fill=black] (3,0) circle (1.5pt);
        \draw[color=black] (3,-0.5) node {$1, S_b$};
        
        \draw [fill=black] (-3,0) circle (1.5pt);
        \draw[color=black] (-3,- 0.45) node {$-1 , S_a$};
        
        \draw [fill=black] (0,0) circle (1.5pt);
        \draw[color=black] (0,-0.5) node {$0$};

        \draw[->]  (0, 0.7) to (0, 0.1);
        \draw[->]  (-3, 0.7) to (-3, 0.1);
        
        \end{tikzpicture}
        \caption{}
        \label{fig:lb-k>m/2-i1}
    \end{subfigure}%
    ~ 
    \begin{subfigure}[t]{0.5\textwidth}
        \centering
        \begin{tikzpicture}[line cap=round,line join=round,>=triangle 45,x=1cm,y=1cm]
        
        \draw [line width=2pt,color=black] (-3,0)-- (3,0);

        \draw (0,1) node[anchor=center] {$|\voters(0)|=\frac{n}{2}$};
        \draw (3,1) node[anchor=center] {$|\voters(1)|=\frac{n}{2}$};
    
        \draw [fill=black] (3,0) circle (1.5pt);
        \draw[color=black] (3,-0.5) node {$1, S_b$};
        
        \draw [fill=black] (-3,0) circle (1.5pt);
        \draw[color=black] (-3,- 0.45) node {$-1 , S_a$};
        
        \draw [fill=black] (0,0) circle (1.5pt);
        \draw[color=black] (0,-0.5) node {$0$};

        \draw[->]  (0, 0.7) to (0, 0.1);
        \draw[->]  (3, 0.7) to (3, 0.1);
        
        \end{tikzpicture}
        \caption{}
        \label{fig:lb-k>m/2-i2}
    \end{subfigure}
    \caption{
    Instance $\instance$ with two metrics, $d_1$ (a) and $d_2$ (b), used to prove the lower bound on the distortion of $k$-winner voting rules for $k > m/2$, as discussed in Theorem \ref{thm:lb->m/2}. $S_a$ and $S_b$ are two subsets of candidates located at positions $-1$ and $+1$. As shown in the figures, voters are divided into two groups of size $n/2$.
    }
\end{figure}

    Given these instances, we prove a lower bound on any voting rule $f$. Let $S = f(\succ)$ be the output committee selected by $f$. Assume that $S$ contains $r$ candidates from $S_a$ and $k - r$ candidates from $S_b$. One can observe that
    \begin{align*}
        \dist(S, \instance; d_1) &= 1 +  \frac{2(m/2 - r)}{3k - m}, \\
        \dist(S, \instance; d_2) &= 1 + \frac{2(r- k + m/2)}{3k - m}.
    \end{align*}
    Since distortion of $f$ can be lower bounded by both of these two values, it can also be lower bounded by their average. Therefore,
    \begin{align*}
        \dist(f) &\ge (\dist(S, \instance; d_1)  + \dist(S, \instance; d_2))/2\\
        &= 1 + \dfrac{m - k}{3k - m},
    \end{align*}
    which completes the proof.
\end{proof}


The results in this section suggest that as the size of the elected committee increases, both the upper and lower bounds indicate the potential for improved distortion. In particular, allowing larger values of $k$ creates opportunities to design voting rules with lower distortion.

\section{Egalitarian Additive Cost}
\label{sec:egalitarian}
In the previous sections, we studied the distortion of multi-winner elections under the utilitarian additive cost objective. The social cost of a committee in the utilitarian setting is calculated by summing the cost of the committee across all voters. In this section, we study the egalitarian objective, focusing on the maximum cost incurred by any voter for the selected committee. This objective captures a notion of fairness.

\cite{cembrano2025metricdistortionpeerselection} have studied the egalitarian additive cost objective on the line metric in peer selection, where the candidates are the voters themselves. They proved a lower bound of $3/2 - 1/k$ in this setting, which is a special case of our setting. We complement this result by proving an upper bound of $2$ in Theorem \ref{thm:egal-ub}.

\begin{theorem} \label{thm:egal-ub}
If an election $\instance = (\voters, \candidates, k, \succ)$ is given, with non-Pareto-dominated candidates ordered as $x_{c_1} \le x_{c_2} \le \dots \le x_{c_m}$ on the line, then for any committee $S$ of size $k < m-1$ excluding $c_1$ and $c_m$, we have $\dist(S) \le 2$ under the egalitarian additive cost objective.\end{theorem}

\begin{proof}
Fix an election $\instance$, a consistent metric $d$, and a committee $S$ such that $c_1, c_m \not \in S$. Let $x_{v_1} \le x_{v_2} \le \dots \le x_{v_n}$ denote the voters' order on the line, and let $S=\{s_1, \dots , s_k\}$, such that $x_{s_i} \le x_{s_{i + 1}}$. 

Recall that the cost of the committee $S$ for a voter $i$ is defined as $\SC(S, i) = \sum_{a \in S} d(i, a)$. 
Since all Pareto-dominated candidates are removed, there can be at most one candidate to the left of $v_1$ and at most one to the right of $v_n$. 
Therefore, if $c_1, c_m \not\in S$, for all $s_i \in S$, we have $x_{v_1} \le x_{s_i} \le x_{v_n} $.

Now, we prove that $\SC(S, \voters) = \max \{\SC(S, v_1), \SC(S, v_n)\}$. Suppose, to the contrary, that there exists a voter $v^*$ such that 
$$\SC(S, v^*) > \max\{\SC(S, v_1), \SC(S, v_n)\}.$$
If all candidates in $S$ are positioned on one side of $v^*$, then it follows that the cost of $S$ for either $v_1$ or $v_n$ would not be less than that of $v^*$. 
Therefore, $v^*$ lies between two consecutive candidates in $S$, say $s_j$ and $s_{j+1}$. Without loss of generality, assume that $j \le k/2$. Moreover, let $I = \{j + 1, j + 2, \dots, k - (j+1) \}$. We have
\begin{align*}
    \SC(S, v_1) - \SC(S, v^*) &=\sum_{i\in I} \left(d(v_1, s_i) - d(v^*, s_i)\right)+ \sum_{i\not \in I} \left(d(v_1, s_i) - d(v^*, s_i)\right) \\
    &\ge  \sum_{i\not \in I} \left(d(v_1, s_i) - d(v^*, s_i)\right) \\
    &= \sum_{i=1}^{j} \Bigl( \bigl(d(v_1, s_i)+d(v_1, s_{k-i})\bigr) - \bigl(d(v^*, s_i) + d(v^*, s_{k-i})\bigr) \Bigr) \\
    &= \sum_{i=1}^{j} \Bigl( \bigl(2d(v_1, s_i)+d(s_i, s_{k-i})\bigr) - \bigl(d(s_i, s_{k-i})\bigr) \Bigr)\\
    &\ge 0.
\end{align*}
Thus, $\SC(S, v^*) \le \SC(S, v_1)$, which is a contradiction. Therefore,
$$\SC(S, \voters) = \max \{\SC(S, v_1), \SC(S, v_n)\} \le \SC(S, v_1) + \SC(S, v_n).$$
Now let $\opt$ be the optimal committee of size $k$, i.e., the committee with the minimum egalitarian additive cost.
We have
\begin{align*}
    \dist(S, \instance) &= \dfrac{\SC(S, \voters)}{\min_{S' \in {\candidates \choose k}} \SC(S', \voters)}\\
    &\le \dfrac{\SC(S, v_1) + \SC(S, v_n)}{\max_{i\in \voters} \SC(\opt, i)}\\
    &\le  \dfrac{\SC(S, v_1) + \SC(S, v_n)}{\frac{1}{2} \bigl(\SC(\opt, v_1) + \SC(\opt, v_n)\bigr)}\\
    &= \dfrac{\sum_{s_i\in S}(d(s_i, v_1) + d(s_i, v_n))}{\frac{1}{2} \sum_{s_j\in \opt}(d(s_j, v_1) + d(s_j, v_n))}\\
    &\le\dfrac{k d(v_1, v_n)}{\frac{1}{2} k d(v_1, v_n)}\\
    &=2.
\end{align*}
The last inequality holds because of the triangle inequality in the denominator of the fraction and the fact that for all $s_i \in S$, $s_i$ lies between $v_1$ and $v_n$. This completes the proof.
\end{proof}


In addition to establishing a lower bound, \cite{cembrano2025metricdistortionpeerselection} introduced a $k$-winner voting rule \textit{$k$-Extremes}, which achieves an upper bound of approximately $3/2$ for peer selection on the line metric under the egalitarian additive cost objective. This voting rule returns the leftmost $\lfloor k/2 \rfloor$ and the rightmost $\lceil k/2 \rceil$ candidates as the output committee.

In Observation \ref{obs:egal-lb-k-extreme}, we show that the distortion of $k$-Extremes is at least $2$ in the general setting.

\begin{observation}\label{obs:egal-lb-k-extreme}
The distortion of the $k$-Extremes voting rule is at least $2$ on the line metric under the egalitarian additive cost objective.
\end{observation}
\begin{proof}
Let $k=2$. In Figure \ref{fig:egal-lb-k-extreme}, candidates $a$, $b$, and $c$ are positioned at $-1$, $0$, and $+1$, with voters $v_1$ and $v_2$ at $0$ and $+1$. The $k$-Extremes rule selects $\{a, c\}$, yielding $\SC(\{a, c\}, \voters) = 2$, while the optimal committee $\{b, c\}$ has $\SC(\{b, c\}, \voters) = 1$. Thus, distortion is at least $2$.

This instance extends to larger $k$ by placing at least $k/2$ candidates at $-1$, $0$, and $+1$. $k$-Extremes selects those at $-1$ and $+1$, yielding a social cost of $k$. Choosing candidates at $0$ and $+1$ instead reduces the social cost to $k/2$, proving a lower bound of $2$.
\end{proof}

\begin{figure}[t]
        \centering
        \begin{tikzpicture}[line cap=round,line join=round,>=triangle 45,x=1cm,y=1cm]
        
        \draw [line width=2pt] (3,0)-- (-3,0);
        
        \draw (3,1) node[anchor=center] {$v_2$};
        \draw (0,1) node[anchor=center] {$v_1$};

        \draw [fill=black] (-3,0) circle (2pt);
        \draw[color=black] (-3,-0.5) node {$-1, a$};
        \draw [fill=black] (3,0) circle (2pt);
        \draw[color=black] (3,-0.5) node {$+1, c$};
        \draw [fill=black] (0,0) circle (2pt);
        \draw[color=black] (0,-0.5) node {$0, b$};

        \draw[->]  (3, 0.7) to (3, 0.1);
        \draw[->]  (0, 0.7) to (0, 0.1);
        
        \end{tikzpicture}\caption{Instance $\instance$ used to prove a lower bound on the distortion of $k$-Extremes in the egalitarian additive costs. The candidates $a$, $b$, and $c$ are located on positions $-1, 0$ and $+1$. 
        Moreover, voters $v_1$ and $v_2$ are located on $0$ and $+1$, respectively.}
        \label{fig:egal-lb-k-extreme}
    \end{figure}

Finding a way to close the gap between $3/2 - 1/k$ and $2$ remains an interesting future direction.

\section{Conclusion}
\label{sec:conclusion}
In this paper, we proposed a new $k$-winner voting rule and established distortion bounds for every committee size $k$ in the line metric under the utilitarian additive cost objective.
 The proposed voting rule achieves optimal distortion for $2$-winner and $3$-winner elections, bounds of $1 + \sqrt{2}$ and $7/3$, respectively, and an upper bound close to $7/3$ for general $k$-winner elections. Moreover, we proved lower bounds on the distortion of any $k$-winner voting rule when $k \leq m/2$. For odd values of $k$, this lower bound is $2 + 1/k$, and for even values of $k$, it is $1 + \sqrt{1 + 2/k}$. We also analyzed the egalitarian additive cost objective and established an upper bound of $2$ in this setting.

We believe our bounds hold under the non-degenerate locations assumption, where voters lie at distinct locations. For both our lower bounds and upper bounds, we can distribute voters and candidates at distinct points near their original locations. There exists a small interval around point $x$, such that by relocating them within this interval the preferences, the optimal output, and the resulting bounds remain unchanged.

However, several open problems remain. The most significant challenge is extending our analysis to more general metric spaces. Understanding how these bounds behave in more complex settings is an important next step. Additionally, even in the line metric, tightening the bounds for general values of $k$ remains an open question.

\pagebreak
\bibliographystyle{plainnat}
\bibliography{References}
\pagebreak

\end{document}